\documentclass[comsoc]{IEEEtran}
\usepackage{cite}
\usepackage{xcolor}
\usepackage{graphicx}

\usepackage[caption=false,font=normalsize,labelfont=sf,textfont=sf]{subfig}
\usepackage{titlesec}

\usepackage{amsmath,amssymb}
\usepackage{tikz, pgf}
\usepackage{pgfplots}
\usepgfplotslibrary{groupplots}
\pgfplotsset{compat=1.18}
\usepgfplotslibrary{fillbetween}
\usetikzlibrary{arrows, arrows.meta, angles, quotes,positioning, calc, shapes.geometric,intersections}

\usepackage{mathtools, cuted}

\newcommand{\jmg}[1]{\textcolor{black}{#1}}
\newcommand{\kz}[1]{\textcolor{black}{#1}}
\newcommand{\fb}[1]{\textcolor{black}{#1}}
\renewcommand{\Pr}{\mathbf{P}}

\newcommand{\Pro}{\mathbf{P}^o_\Psi}
\newcommand{\Pry}{\mathbf{P}^y_\Psi}

\newcommand{\E}{\mathbf{E}}
\newcommand{\Eo}{\mathbf{E}^o_\Psi}

\newcommand{\N}{\mathbb{N}}
\newcommand{\R}{\mathbb{R}}
\renewcommand{\o}{\circledcirc}
\newcommand{\SIC}{\ensuremath{\mathrm{SIC}(\eta_0)}}

\newcommand{\SINR}{\ensuremath{\mathrm{SrINR}}}
\newcommand{\SIR}{\ensuremath{\mathrm{SrIR}}}

\newtheorem{theorem}{Theorem}
\newtheorem{definition}{Definition}
\newtheorem{remark}[theorem]{Remark}
\newtheorem{corollary}[theorem]{Corollary}
\newtheorem{lemma}[theorem]{Lemma}
\newtheorem{proposition}[theorem]{Proposition}

\usepackage{hyperref}

\begin{document}

\title{Decoding Delay Guarantees of Space Regulated\\Multiple Access Random Wireless Networks\\using Successive Interference Cancellation}

\author{Kevin Zagalo,~Jean-Marie Gorce~and~François Baccelli%
\thanks{This work was supported by the ERC NEMO grant, under the European Union’s Horizon 2020 research and innovation program, grant agreement number 788851 to INRIA, and the France 2030 projects PEPR réseaux du Futur project (grant ANR-22-PEFT-0010). The work of F. Baccelli was supported by the Horizon Europe INSTINCT project (grant SNS 101139161).}%
\thanks{Kevin Zagalo is with Inria, INSA Lyon, CITI, UR3720, 69621 Villeurbanne, France (e-mail: kevin.zagalo@inria.fr).}%
\thanks{Jean-Marie Gorce is with Inria, INSA Lyon, CITI, UR3720, 69621 Villeurbanne, France.}%
\thanks{François Baccelli is with INRIA-ENS, 75012 Paris, France and Telecom
Paris, 91120 Palaiseau, France.}%
}%

\maketitle

\begin{abstract}
    This paper is focused on decoding delay guarantees in wireless networks, where messages have a given signal-to-interference-plus-noise ratio  threshold $\eta_0$ to meet in order to be successfully decoded, and where transmissions should occur within some strict time constraints. Its main contribution consists in quantifying the worst-case transmission decoding delays in the uplink of cellular and cell-free networks using successive interference cancellation. We show how such decoding delay guarantees can be obtained using spatial network calculus, a new tool introduced recently, and in particular spatial regulation. The results rely on the assumption of absence of fading. We nevertheless outline what this approach will lead to in the fading case for cellular networks.
\end{abstract}


\section{Introduction}


\IEEEPARstart{W}{ith} the growth of the Internet of Things and the evolution towards 5G/6G, massive connectivity, 
low power consumption, and low latency, are key requirements to build an Ultra-Reliable Low Latency Communication
(URLLC) network. URLLC aims to provide guarantees for mission-critical applications, 
from factory automation to safety control \cite{popovski2019wireless}. It prioritizes high levels of reliability and minimal latency,
to send (or receive) time-sensitive and critical information to (from) a massive set of devices.
Whether in the downlink or in the uplink, URLLC must provide deterministic guarantees on both reliability
and latency of communications. In the multi-user context, the performance analysis is complex, as multiple
transmitters concurrently access the shared communication channel, leading to the need for an
efficient interference cancellation or control strategy. 

Decoding delay analysis consists in finding a time slot size, or \textit{span}, $T$
and a suited scheduling algorithm (e.g., \cite{tuninetti2018scheduling,sheldon2021gaussian,nikbakht2023broadcast}),
such that a predefined set of transmissions can be completed before $T$ units of time. 
In this paper, we specifically focus on \fb{a multi-user information-theoretic decoding} 
of the transmission. More precisely, we want to provide guarantees that, in the worst-case, almost surely,
all messages sent from transmitters to receivers are decoded in no more than $T$ channel uses for a suited span $T$. We consider a set of transmitters
(user nodes) and a set of receivers (base stations or aggregators), in the uplink. The analysis is given for a typical transmission
(from typical transmitters to a typical receiver), using Palm Calculus \cite{Bremaud2020}.
The general purpose of this work is to analyze how the geometry of the network impacts the decoding delays and the span $T$
of the typical uplink transmission in the context described above, and how to choose the span of the network. \fb{The decoding is based on Successive Interference Cancellation (SIC).}

\subsection{Context}
Mitigating interference is a key challenge in multi-user communication scenarios,
as interference arising from nearby transmissions can degrade the transmission 
and reduce overall system performance. We analyze how the statistical
properties of the geometry of the wireless network impact uplink decoding delay performance. We focus on
SIC, which provides a scheduled
order of decoding, with the closest transmitter being decoded first. 
The space-time interference correlation is essential in this context.
This question is widely discussed in \cite{9516701} using Poisson point processes (PPP).
These Poisson assumptions are fundamentally incompatible with deterministic guarantees.
This explains why we use the spatial network calculus (SNC) framework rather than the Poisson framework 
in this paper. This allows us to provide a deterministic bound on the decoding delays associated with this dynamic.
Note that many papers consider the (local) decoding delay of transmissions in terms of number of attempts
for a successful transmission, accounting for the geometry of the  network
\cite{baccelli2020random, danufane2023analysis, haenggi2012local}. In these papers, this is again done
in the Poisson setting. In contrast, in our work, the decoding delay is the number of channel uses necessary
for a subset of simultaneous transmissions to be successfully decoded. We introduce a delay analysis based on space regulation, built upon
the regulated class introduced in \cite{feng2023spatial}, to provide deterministic guarantees on the reception.

\subsubsection{Space regulation}
The concept of space regulation is quite recent. Introduced in \cite{feng2023spatial},
it generalizes the Network Calculus (NC) framework \cite{bouillard2018deterministic} to multidimensional point processes. By imposing distance-dependent constraints on active transmitters, spatial regulation deterministically bounds interference from above, ensuring corresponding lower bounds on the signal-to-interference-plus-noise ratio (SINR) of a network. 
Usually, PPPs are
used to model transmitters and/or receivers, e.g. 
\cite{baccelli2020random, blaszczyszyn2013using, haenggi2012local, danufane2023analysis, zhang2014performance},
due to their analytical tractability that allows one to express many important quantities in a closed form.
However, in order to control interference and coverage, having transmitters or receivers
that are potentially very close to each other, without bounds on their number, prevents \fb{any deterministic guarantee} on reliability or latency. One can however build a spatially regulated process
from a PPP, by adding a spatial regulation on transmissions, c.f., \cite{feng2025performance}. 

\subsubsection{Delay analysis}
The approach to characterize the network traffic when interference is treated as noise is 
discussed in, e.g., \cite{telatar1995combining} and \cite{ephremides1998information}. 
Suppose there are $k$ active transmitters. 
 The power received at receiver $y$ from transmitter $x$ is $P_{xy}$.
The SINR of this link is hence given by $\eta_{xy} = \frac{P_{xy}}{\sum_{i=1}^k P_{x_i y} + N_0}$ for some power of thermal noise $N_0$, where $\{x_1,\dots,x_k\}$ is the subset of transmitters interfering with the $xy$ link.
According to the Shannon \fb{channel coding} theorem, the ergodic transmission rate $R_{xy}$ of the link $xy$ is upper-bounded by the
information-theoretic capacity limit, namely $R_{xy} \leq \log_2(1 +\eta_{xy}) ~ \textrm{bpcu}$. This capacity gives us the maximum rate at which the messages sent by transmitters to receivers can be decoded. 
Considering a subset of transmitters sending independent messages of $m$ bits each, the minimal transmission time $T$ allowing to decode all messages is given by $T=\lceil \frac{m}{R} \rceil$ channel uses, where $R = \min_{x, y} R_{xy}$ is the smallest rate possible in the set of reliable links, which we determine in this paper. Hence the span $T$ can be seen as the duration to transmit all messages. 
In what follows, we look at
how the geometric properties of the network impact the span $T$. 
\subsubsection{Cell-free networks} In this paper, we also consider a cell-free (CF) network, 
where all receivers listen to all transmitters and where
the interference on a given transmission stems from all other transmissions in the network.
In CF networks, receivers jointly attempt to decode the transmitter's messages. Starting in Section~\ref{sec:CF}, we consider that, when a receiver 
succeeds in decoding a given transmitter, all the other receivers are instantaneously informed through some
backhaul network. 
CF networks have been analyzed in a geometry-based perspective \cite{8972478,8379438}, but not yet
from a latency point of view to the best of our knowledge. This paper can hence be seen as a first step into latency-focused analysis of CF networks. No cooperation between receivers is assumed until Section~\ref{sec:CF}. 

\subsection{Motivations and contributions}

As already stressed, the Poisson model is incompatible with deterministic guarantees, as the number
of interfering transmissions cannot be almost surely bounded from above. This impossibility leads
to the consideration of other classes of point processes (PP).
In this paper, we analyze the latency of the decoding of messages when assuming
that the locations of the transmitters $\Phi$ and the receivers $\Psi$ in the plane satisfy certain spatial regulation properties. We consider a stationary transmitter PP $\Phi$, and a jointly stationary receiver PP $\Psi$ in $\R^2$: firstly communicating through a cellular network, see Section~\ref{sec:time}; and in a second part we consider a CF network, see Section~\ref{sec:CF}. We then seek an upper-bound on the decoding delay $D(x,\Phi)\in\R_+$ for any transmitter located at \fb{$x\in\Phi_0$, where $\Phi_0$ is a jointly stationary sub-PP of $\Phi$ which is explicitly defined in \eqref{eq:initPP}, and which corresponds to the set of transmitters in $\Phi$ with a SINR with respect to (w.r.t.) the receivers in $\Psi$ above some threshold }. This will imply that the transmissions of messages of size $m$ bits sent by $x\in\Phi$ can be decoded using SIC by the typical receiver in less than $D(x,\Phi)$ channel uses almost surely for the given spatial regulation.
Equivalently, the probability that the \textit{worst-case decoding delay} of the typical receiver using SIC is almost surely upper-bounded for any realization of the transmitters set $\Phi$, i.e., $$\sup_{x\in\Phi_0} D(x,\Phi) \leq \frac{m}{\log_2(1 + \eta_0)} , \quad \Pro-a.s.,$$ 
in the absence of fading. 
In Section~\ref{sec:fading}, we extend what this approach leads to the fading case. 
Finally, in Section~\ref{sec:CF}, we provide in Proposition~\ref{prop:CFSINRBound} conditions on $\eta_0$ for the worst-case delay of CF networks $D(\Phi, \Psi)$ defined in \eqref{def:CFdelay}, to be lower or equal to $\frac{m}{\log_2(1 + \eta_0)}, \ \Pr$-almost-surely in the CF case.

\subsection{Model} \label{sec:sicmodel}

    \begin{table}[t]
        \caption{Notations}
        \begin{tabular}{c|p{0.75\linewidth}}
            $\Phi$ & set of transmitters  \\
            $\Psi$ & set of receivers  \\
            $\Pro$ & Palm probability of $\Psi$ with $o$ as typical receiver \\
            $\Pr^o_\Phi$ & Palm probability of $\Phi$ with $o$ as typical transmitter \\
            $\mathrm{Vol}$ & Lebesgue measure on $\R^2$ \\
            $N_0$ & noise density \\
            $P_0$ & emission power \\
            $\gamma_0$ & inverse of signal-to-noise ratio at the origin $N_0/P_0$\\
           $\SINR(x,\Phi)$ & signal-to-residual-interference-plus-noise ratio from $x\in\Phi$ to the typical receiver\\
           $\SIR(\Phi, \Psi)$ & signal-to-residual-interference ratio from $\Phi$ to $\Psi$\\
           $m$ & message size \\
            $R(x,\Phi)$ & rate from $x\in\Phi$ to the typical receiver\\
           $\eta(r)$ & SrINR bound at distance $r$ \\
            $\eta_0$ & SrINR threshold \\
            $\ell$ & path-loss function\\
            $\tau(\Phi)$ & coverage distance of the typical receiver for the transmitters set $\Phi$\\
            $\tau_0$ & coverage distance lower-bound\\
            $T$ & span, time slot size  \\
            $d(x, \Phi)$ & virtual decoding delay from $x\in\Phi$ to the typical receiver\\
            $D(x,\Phi)$ & decoding \SIC-delay from $x\in\Phi$ to the typical receiver\\
            $D(\Phi, \Psi)$ & decoding \SIC-delay from $\Phi$ to $\Psi$ in the CF case
         \end{tabular}
        \label{tab:notations}
    \end{table}
Let $\R^2$ denote the Euclidean plane, $b(x,r)$ the open ball of $\R^2$ centered at $x\in\R^2$ of radius $r > 0$,
and $B(x,r)$ its closure,
and $\o_y(r, R) = b(y, R) \setminus B(y,r)$ the open ring centered at $y\in\R^2$, of inner radius $r > 0$
and outer radius $R > r$. Let $o = (0,0)$
and $\o(r, R) = \o_o(r, R)$. Let also $x^+ = \max(0,x)$. Let $f$ and $g$ be two functions from $\R_+$ to $\R_+$. We write $f(x) = O(g(x))$ if there exist $C > 0$ such that $\forall x > 0$, $f(x) \leq Cg(x)$; $f(x) = \Omega(g(x))$ if there exist $C > 0$  such that $\forall x > 0$, $f(x) \geq Cg(x)$; and $f(x) = \Theta(g(x))$ if $f(x) = \Omega(g(x))$ and $f(x) = O(g(x))$. We denote as $\triangleq$  an equality understood as a definition.
Notation is summarized in Table~\ref{tab:notations}.
The parameters used in simulations throughout the paper are summarized in Table~\ref{tab:simulations}.
\subsubsection{Spatial model} The transmitters and receivers PPs $\Phi$ and $\Psi$ are assumed to be stationary and defined on the
probability space $(\Omega, \mathcal{F}, \Pr)$. $\Pro, \Eo$ respectively denote the Palm probability
and the Palm expectation of $\Psi$ with $o$ as a typical receiver, and  $\Pr^o_\Phi, \E^o_\Phi$ respectively denote the Palm probability and the Palm expectation of $\Phi$ with $o$ as a typical transmitter \cite{Bremaud2020}. We suppose that the couple $(\Phi, \Psi)$ is jointly stationary, i.e.,
for any receiver $y\in\Psi$, $\Pry ( \Phi \in A) = \Pro(\theta_y \Phi \in A)$, c.f. \cite[Theorem~8.4.11.]{Bremaud2020}, with $\theta_y\Phi = \sum_{x\in\Phi} \delta_{x-y}$.
Note that, as explained in \cite[Remark 7]{feng2023spatial}, if $\Phi$ and $\Psi$ are independent, then 
$\Pro(\Phi \in A) = \Pr(\Phi \in A)$. 


\subsubsection{Slotted time model}
We assume throughout the paper that time is slotted, with slots of size $T\in\N$, that we determine in the following. By this we mean that during a time interval of size $T$ units of time, only one transmission of $m$ bits per transmitter is sent to the receivers.
We suppose that \jmg{all receivers try to decode messages simultaneously}, and that each receiver is aware of the SINR associated to each \jmg{detected} 
transmission before decoding. We call this assumption the \textit{synchronous arrival case}. This assumption is justified if we consider that the simultaneous
transmissions are preceded by a signaling message allowing to detect when messages are transmitted,
as extensively studied, e.g., \cite{chetot2024hybrid}. Hence, at each instant of the form $kT, k\geq 1$, all active transmitters send a message of size $m$. 

\subsubsection{Channel model}
We consider additive white Gaussian noise (AWGN) channels without fading, only dependent on the transmission power, the channel noise and distance.
The power path-loss function is denoted by $\ell(r)$, with $\ell(\|x-y\|)$ representing the path-loss
between transmitter $x\in\Phi$ and receiver $y\in\Psi$. Throughout the paper we assume that the path loss is such that $\ell(r) = \Theta(r^{-\beta}), \beta>2$, which means that  $\ell(r)$ behaves like $Cr^{-\beta}$, for some $C > 0$. We make the reasonable assumption that $\ell(0) \leq 1$. 
We suppose the transmission power to be
constant for all transmitters and denote it by $P_0$.  Thus the power received
at $y\in\R^2$ from $x\in\R^2$ is $P_{xy} = P_0\ell(\|x-y\|)$. 
Let $N_0$ be the power of the thermal noise. Let us denote $\gamma_0\triangleq N_0/P_0$. \jmg{Until Section~\ref{sec:fading}, no fading is considered; in the fading case, the fading distribution is assumed to admit exponential moments.} 

\subsubsection{Communication system model}\label{sec:sysmodel}
In our system model, a transmitter transmits at the beginning of a timeslot of size $T\in\N$. 
During this slot, it transmits successive code words according to predefined codebooks. 
The transmission scheme behind this formulation is variable-length sparse feedback codes, as suggested  in \cite{yang2022variable,yavas2023variable}. Let us describe the main idea of these coding schemes for a point to point transmission: consider a source that wants to transmit a \jmg{message} $s$, over $T$ successive channel uses. 
The transmitter and receiver agreed on a codebook, randomly chosen, such that at the $k$-th timeslot, the transmitter sends \jmg{the code sequence  $c(s)$}  mapped to a base-band signal. 
At each channel use, the decoder estimates $\hat s$ from $c(s)$, itself from the demodulation of the received signal through an AWGN channel.  Using Gaussian codebooks, we know that the Shannon rate is achievable for the maximum likelihood decoder. As we consider no fading, the power of the point-to-point signal at distance $r > 0$ is $P_0 \ell(r)$. Assuming now that two transmissions occur simultaneously from distance $r_1$ (resp. $r_2$) to a receiver, then the power of the signal at the receiver is $P_0(\ell(r_1) + \ell(r_2))$ and the SINR associated with distance $r_1$ (resp. $r_2$) is $\frac{\ell(r_1)}{\ell(r_2) + \gamma_0}$ (resp.$\frac{\ell(r_2)}{\ell(r_1) +  \gamma_0}$). In order to mitigate this signal superposition, we use a Successive Interference Cancellation (SIC) scheme that we describe more precisely in the next section. Using for example a maximum likelihood decoder, the receiver can evaluate the error probability of its decoding decision and if this error is lower than some threshold, it accepts the decision and sends a feedback signal. 
Asymptotically, when $T$ is large enough, the rate of this model converges to the capacity of the point-to-point channel, which means that, at channel use $T$, the decoder can reliably decode the source $s$ with an arbitrarily low error probability, as long as the capacity is larger than $m = \Omega(\log_2 M)$, with $M$ the size of the codebook. Nevertheless, in \cite{yavas2023variable}, the authors evaluate the non-asymptotic regime. In our paper, we consider the asymptotic (or ergodic) case, with $T$ large enough, and we rather focus on spatial properties of the transmitters.
But the coding scheme is similar. 
In our system we consider the same transmission scheme $s^x \to c(s^x)$, for each message $s^x$ of each transmitter located at $x\in\Phi$ transmitting simultaneously. 
We do not consider any feedback; we consider instead a fixed span $T$, such that we can guarantee that every transmission can be reliably decoded. With the \SIC~model introduced in the following section, once a message can be decoded, the signal from this transmitter can be fully cancelled. 
All transmitters transmit until the end of the time slot. \kz{A stop feedback code is then sent at time $T$ by all receivers}.

\subsubsection{The \SIC~model} 

Under \SIC, messages are decoded in decreasing order according to their received power. In our setup, this translates to decoding in increasing order according to distance: a message from a transmitter $x\in\Phi$ will be decoded by receiver $y\in\Psi$ only after all transmissions with higher received power, i.e. those from transmitters located in $\Phi \cap b(y,\|x-y\|)$, have been successfully decoded. In this framework, we consider two fundamental types of delays: the \textit{virtual decoding delay}, denoted as $d(x, \Phi)$, which is the number of channel uses necessary for a transmission to be decoded without accounting for the cancellation of any higher power transmissions; and the \textit{\SIC-decoding delay}, denoted as $D(x, \Phi)$ (formally defined in Definition~\ref{def:SICdelay}), which accounts for the \SIC~decoding process at the typical receiver. In the case of CF networks, the \SIC-decoding delay of the typical transmitter located at the origin is denoted as $D(\Phi, \Psi)$ and accounts for all other receivers that can reliably decode a transmission, as formally defined in \eqref{def:CFdelay}. In order to be reliably decoded, a transmission must have a SINR greater than or equal to a threshold $\eta_0$; However, the SINR under SIC is not the standard SINR, because it considers only the lower received power signals as interference to a given transmission. We call this SINR the signal-to-residual-interference-plus-noise (\SINR). The residual interference was first introduced in \cite[Definition~3]{zhang2014performance}. Because receivers decode iteratively from highest to lowest power signals, the worst-case (largest) decoding time is bounded from above by considering the worst (smallest) \SINR~over all transmitters, \jmg{which is not necessarily the lowest received power signal.} Indeed, \SIC~considers only the lower received power signals as \textit{residual interference} to a given transmission. This concept of residual interference was first introduced in \cite[Definition~3]{zhang2014performance}. Receivers using \SIC~and variable length coding proceed as follows: the receiver first attempts to decode the message from the transmitter with the highest received power signal under full interference at time $d_1$ (the virtual decoding delay). When this message is decoded, which happens in finite time if this first \SINR~is less than $\eta_0$, the corresponding signal can be cancelled from the signal of the first $d_1$ channel uses. Then the second highest received power signal is decoded using the same scheme at time $\max\{d_1, d_2\}$. This continues iteratively until receivers meet a \SINR~below the threshold $\eta_0$. See \figurename~\ref{fig:sic}. While this successive cancellation process induces overheads, we ignore their impact and assume that \SIC~is capable of isolating every transmission even if received powers are arbitrarily close. Finally, we assume that once a message is cancelled from the aggregate signal, it is immediately decoded, thus we ignore decoding errors.

\begin{figure}[t]
    \centering
    \begin{tikzpicture}[scale=0.65]

    \draw[fill=gray!10, draw=gray!40, thick] (0,0) rectangle (10.5,0.7);
    \node[gray!80!black, font=\small] at (1.5, 0.35) {Noise};
    
    \draw [{Stealth}-{Stealth}, gray!80!black] (10.2, 0.05) -- (10.2, 0.65);
    \node[gray!80!black, font=\small, anchor=east] at (10.2, 0.35) {$N_0$};

    \draw[fill=green!10, draw=green!60!black, thick] (0,0.7) rectangle (10.5,1.7);
    \node[green!60!black, font=\small] at (2.2, 1.2) {\textbf{Transmitter 3}};
    
    \draw [{Stealth}-{Stealth}, green!60!black] (10.2, 0.75) -- (10.2, 1.65);
    \node[green!60!black, font=\small, anchor=east] at (10.2, 1.2) {$P_3$};

    \draw[fill=orange!15, draw=orange!60, thick] (0,1.7) rectangle (10.5,3.2);
    \node[orange!80!black, font=\small] at (2.5, 2.45) {\textbf{Transmitter 2} };
    
    \draw [{Stealth}-{Stealth}, orange!80!black] (7.7, 1.75) -- (7.7, 3.15);
    \node[orange!80!black, font=\small, anchor=east] at (7.7, 2.45) {$P_2$};

    \draw[fill=blue!15, draw=blue!60, thick] (0,3.2) rectangle (10.5,5.7);
    \node[blue!80!black, font=\small] at (2.75, 4.45) {\textbf{Transmitter 1} };
    
    \draw [{Stealth}-{Stealth}, blue!80!black] (5.2, 3.25) -- (5.2, 5.65);
    \node[blue!80!black, font=\small, anchor=east] at (5.2, 4.45) {$P_1$};

    \draw[draw=black, fill=white, thick] (0.5, 6.2) rectangle (5.7, 7.8);
    \node[font=\small] at (3, 7) 
        {$\ \frac{P_1}{P_2 + P_3 + N_0} > \frac{P_3}{N_0} > \frac{P_2}{P_3 + N_0}$};

    \draw[draw=gray!60, fill=gray!5, rounded corners, dashed] (6.0, 6.2) rectangle (12.2, 7.8);
    \node[anchor=west, font=\small, align=left] at (6.1, 7) {
        \tikz[baseline=-0.5ex]\draw [-{Stealth[scale=1.0]}, line width=1pt, red!70!black] (0,0) -- (0.4,0); SIC Subtraction ($\hat{s}_i$) \\
        \tikz[baseline=-0.5ex]\draw [-{Stealth[scale=1.0]}, line width=1pt, blue!70!black] (0,0) -- (0.4,0); Waiting cancellation
    };

    \draw[dashed, gray!80, thick] (5.5,0) -- (5.5,5.7);
    \node[below=4pt, font=\bfseries] at (5.5,0) {$d_1$};
    
    \draw[dashed, gray!80, thick] (7,0) -- (7,1.7);
    \node[below=4pt, font=\bfseries] at (7,0) {$d_3$};
    
    \draw[dashed, gray!80, thick] (8.5,0) -- (8.5,3.2);
    \node[below=4pt, font=\bfseries] at (8.5,0) {$d_2$};
    \node[below=4pt, font=\bfseries] at (0,0) {$0$};
        
    \draw [-{Stealth[scale=1.2]}, thick] (0,0) -- (0,7.8);
    \node[rotate=90, font=\small, anchor=south] at (-0.4, 3.9) {Received Power};

    \draw [-{Stealth[scale=1.2]}, line width=1.2pt, red!70!black] (5.5, 5.7) -- (5.5, 3.2);
    \draw [-{Stealth[scale=1.2]}, line width=1.2pt, red!70!black] (8.5, 3.2) -- (8.5, 1.7);
    \draw [-{Stealth[scale=1.2]}, line width=1.2pt, red!70!black] (8.5, 1.7) -- (8.5, 0.7);

    \draw [-{Stealth[scale=1.2]}, line width=1.2pt, blue!70!black] (7, 1.7) -- (7, 0.7);

    \draw [-{Stealth[scale=1.2]}, thick] (0,0) -- (11.5,0) node[below, align=left, font=\small] {Channel\\Uses};

\end{tikzpicture}
    \caption{Example of SIC when $P_1 > P_2 > P_3$ and $\frac{P_1}{P_2 + P_3 + N_0} > \frac{P_3}{N_0} > \frac{P_2}{P_3 + N_0} > \eta_0$ : Transmitter $i$ sends the signal $s_i$; receiver receives signal with power $P_i$, decodes it as $\hat s_i$ and substracts it from the total signal at time $\max_{j\leq i} d_j$, $i=1,2,3$, where $d_j$ is the time when transmitter $j$ can be decoded if no other transmitter with higher received signal power is being decoded.  One can see here that the \SIC-delay of user $3$ is $d_2$, because transmitter $3$ has to wait transmitter $2$ to be decoded since $P_3 < P_2$. Note that we have assumed that all transmitters $1,2$ and $3$ transmit until the end of the time slot. Thus, their signal being cancelled respectively at instants $d_1, \max\{d_1, d_2\}, \max\{d_1, d_2, d_3\}$, they stop interfering with other transmissions at those respective instants, since their signal has been canceled.
}
    \label{fig:sic}
\end{figure}

\section{Space regulation} \label{sec:reliability}

SNC was introduced in \cite{feng2023spatial}. It generalizes 1-dimensional NC time guarantees to 2-dimensional guarantees.
This section discusses an extension of the \textit{ball regulation} introduced in \cite{feng2023spatial} to
\textit{ring regulation}, and a generalization of \textit{shot-noise regulation}.
It then recalls the notion of \textit{void-regulation}.
It finally discusses the use of this framework for \SIC.

All simulations provided in this paper use Matérn hardcore point processes of type II (MHPP) for the transmitters PP, where points from a stationary PPP of intensity $\lambda$ are kept only if they are at distance at least $2H > 0$. 

\subsection{Ball, ring and shot-noise regulation}

\begin{definition}[$(\sigma, \rho, \nu)$-ball regulation]
Let $\Phi$ and $\Psi$ be two jointly stationary PPs on $\mathbb{R}^2$ and
$(\sigma, \rho, \nu) \in \mathbb{R}^3_+$. $\Phi$ is said $(\sigma, \rho, \nu)$-ball regulated w.r.t. $\Psi$ if 
$\forall R \geq 0$, 
\begin{equation} \Phi(b(o,R)) \leq \sigma + \rho R + \nu R^2, \quad \Pro-a.s.. \label{eq:ball}\end{equation} 
\end{definition}

\begin{definition}[$(\sigma, \rho, \nu)$-ring regulation]
    Let $\Phi$ and $\Psi$ be two jointly stationary PPs on $\mathbb{R}^2$ and
$(\sigma, \rho, \nu) \in \mathbb{R}^3_+$. $\Phi$ is said $(\sigma, \rho, \nu)$-ring regulated w.r.t. $\Psi$ if $\forall R \geq 3r \geq 0$, 
\begin{equation*} \Phi(\o(r,R)) \leq \sigma + \rho (R-r) + \nu (R^2-r^2), \quad \Pro-a.s. \end{equation*} 
\end{definition}

Note that, in \cite{feng2023spatial}, this regulation is called \textit{weak} space regulation. \textit{Strong} space regulation is a regulation w.r.t. the whole Euclidean plane $\R^2$. We use weak space regulation throughout this paper, and refer to strong space regulation as simply space regulated.
The reason why ring regulation is defined for $R \geq 3r \geq 0$ can be found in the proof of the following lemma. 

\begin{lemma}\label{lem:equivalence2}The three following properties hold :
    \begin{enumerate}
        \item If a PP is $(\sigma, \rho, \nu)$-ball regulated, it is also $(2\pi\sigma, \pi\rho, \pi\nu/2)$-ring regulated,
        \item If a PP is $(\sigma, \rho, \nu)$-ring regulated, it is also $(\sigma, \rho, \nu)$-ball regulated,
        \item If a PP is $(\sigma, \rho, \nu)$-ring regulated w.r.t $\Psi$, it is also $(\sigma, \rho, \nu)$-ball regulated w.r.t. $\Psi$.
    \end{enumerate}
\end{lemma}

\begin{IEEEproof} See Appendix~\ref{proof:equivalence2}.
\end{IEEEproof}

This last Lemma shows that ball regulation implies ring regulation only for large enough rings. It does not hold in general that ball regulation implies ring regulation. Also note that, Lemma~\ref{lem:equivalence2} only holds for strong regulation. Nevertheless, it is easy to check that ring regulation implies ball regulation.

        \begin{remark}\label{rmk:hardcore}
Any hardcore process with core distance $H > 0$ is 
$\left(\frac{\pi^2}{\sqrt{3}}, \frac{2\pi^2}{H \sqrt{3}}, \frac{\pi^2}{H^2 \sqrt{3}}\right)$
ring regulated. 
Let $R \ge 3r \geq 0$, we define a midline ball $b(c(\theta) , \delta)$ of radius $\delta = \frac{R-r}{2}$ and centered at $c(\theta) = (y_1 + \frac{R+r}{2} \cos\theta, y_2 + \frac{R+r}{2} \sin\theta)$. With the same argument as in the proof of Lemma~\ref{lem:equivalence2}, the point count is bounded by the rotation integral $\Phi(\circledcirc_y(r,R)) \le \int_0^{2\pi} \Phi(b(c(\theta) , \delta)) d\theta$. Using the hexagonal packing density $\Delta = \pi/\sqrt{12}$ \cite{segre1944densest}, the number of points in any ball $b(c(\theta) , \delta)$ is bounded by the area of the thickened ball $b(\delta + H)$:
$\Phi(b(c(\theta) , \delta)) \leq \Delta\frac{\mathrm{Vol}(b(c(\theta) , \delta+H))}{\mathrm{Vol}(b(c(\theta), H))} = \frac{\pi}{H^2 \sqrt{12}} \left( \delta+H \right)^2$.
Integrating over $\theta \in [0, 2\pi]$ and expanding the quadratic yields
$\Phi(\circledcirc_y(r,R)) \leq \frac{\pi^2}{H^2 \sqrt{3}} \left( (R-r)^2 + 2H(R-r) + H^2 \right)$
and applying $(R-r)^2 \le R^2 - r^2$ leads to the result for $R \geq 3r$. 
        \end{remark}

\begin{remark}\label{rmk:lattice}
    Any lattice, or lattice perturbed by an i.i.d. bounded perturbation $\xi$, is ball regulated. More specifically, a perturbed squared lattice $\Psi$ with intensity $\lambda$ is $\left(\lambda \pi (\delta + \frac{1}{\sqrt{2\lambda}})^2, 2 \lambda \pi (\delta + \frac{1}{\sqrt{2\lambda}}), \lambda \pi\right)$-ball regulated. 
Indeed, if $z + \xi_z \in b(o, R)$, then $\|z\| \leq R + \delta$. Let $V_z$ be the Voronoi cell of $z$. Since the distance from any point in $V_z$ to $z$ is at most $\frac{1}{\sqrt{2\lambda}}$, we have $\bigcup_{z \in \Psi \cap b(o, R+\delta)} V_z \subseteq b(o, R + \delta + \frac{1}{\sqrt{2\lambda}})$. Comparing the Lebesgue measure of both sides, we get 
$\Psi(b(o, R+\delta)) \frac{1}{\lambda} \leq \pi (R + \delta + \frac{1}{\sqrt{2\lambda}})^2.$
It follows that 
$\Phi(b(o, R)) \leq \lambda \pi (R^2 + 2R(\delta + \frac{1}{\sqrt{2\lambda}}) + (\delta + \frac{1}{\sqrt{2\lambda}})^2)$.
\end{remark}

\begin{definition}[$(\sigma, \rho, \nu)$-shot-noise ring regulation]
Let $\Phi$ and $\Psi$ be two jointly stationary PPs on $\mathbb{R}^2$.
$\Phi$ is $(\sigma, \rho, \nu)$-shot-noise ring regulated w.r.t.  $\Psi$ if for all $R\geq 3r> 0$,
\begin{equation}\sum_{x \in \Phi \cap \o(r,R)} \!\!\!\!\ell(\|x\|) \leq \sigma\ell(r)+ \rho \!\int_{r}^{R}
\!\ell(s)ds + 2\nu \!\int_{r}^{R}\!s\ell(s)ds,  \label{eq:shotnoise}\end{equation}
$\Pro$-almost-surely, for all non-negative, bounded and non-increasing functions $\ell$ such that $\int r\ell(r)dr < \infty$.
\end{definition}

    \begin{table}[t]
        \caption{Parameters used in simulations}
        \vspace{.3em}
        \begin{tabular}{c|p{.65\linewidth}}
            Transmitter set & $\Phi$-MHPP of intensity 1 and core distance 3\\
            Receiver set & $\Psi$-square lattice process spaced by $\tau_0/\sqrt{2}$ and displaced by an uniform distribution in $b(o,\tau_0/2)$\\
            Reference distance & $r_0 = 1$ m \\
            Path loss & $\ell(r) = \max\left\{1, \frac{r}{r_0}\right\}^{-\beta}$\\
            Message size & $m = 10^3$ bits\\
           \SINR~threshold & $\eta_0 = -10$ dB
         \end{tabular}
        \label{tab:simulations}
    \end{table}




\begin{lemma}[Equivalence]\label{lem:equivalence}
    Let $\Phi$ and $\Psi$ be two jointly stationary PPs on $\mathbb{R}^2$. $\Phi$ is $(\sigma, \rho, \nu)$-ring regulated w.r.t. $\Psi$ if and only if it is $(\sigma, \rho, \nu)$-ring-shot-noise regulated w.r.t. $\Psi$.
\end{lemma}
    
\begin{IEEEproof}
     See Appendix~\ref{appendix:equivalence}.
\end{IEEEproof}

\subsection{Void and coverage regulation}
Another type of space regulation, called \textit{void regulation} \cite{feng2023spatial},
is needed in order provide guarantees for the whole network. More specifically, this space regulation ensures that at least one receiver can receive the message of any transmitter with a received power of at least $P_0\ell(\tau)$. In this section, we extend the notion of void regulation into coverage regulation.

\begin{definition}[$\tau$-void regulation \cite{feng2023spatial}]
Let $\Psi$ be a stationary PP. We say that $\Psi$ is $\tau$-void 
regulated if $\Psi(b(o,\tau))\geq 1, \quad \Pr$-almost-surely.
\end{definition}

Note that we do not use the void regulation w.r.t. the transmitter PP in this paper, in opposition to ring regulation w.r.t. the receiver PP used in the previous section. Also note that, using the stationarity of $\Psi$ yields to $\Pr(\Psi(b(x,\tau))\geq 1)=1$ for all $x\in\R^2$.

\begin{remark}
    A lattice PP perturbed with an i.i.d. bounded displacement is void regulated, c.f. \cite[Example~2]{feng2023spatial}. 
More specifically, let $\Psi$ be a square lattice with spacing $b>0$ and $\Psi^\prime = \{x + \xi_x : x\in\Psi\}$ 
be its displacement, where $(\xi_x)_x$ is i.i.d. uniform on $b(o,a), a>0$. 
Let $y = x + \xi_{x} \in \Psi^\prime, x \in \Psi$. Let $V_x$ be the Voronoi cell with nucleus $x$. Then for any $z\in V_x$, the triangular inequality yields
\begin{equation*}
\|y - z\| \leq \| x + \xi_x - z\| \leq \|x - z\| + \|\xi_{x}\| \leq b/\sqrt{2} + a.
\end{equation*} 
Hence $\Psi^\prime$ is $\tau$-void regulated with $\tau= b/\sqrt{2} + a$.
We use this property in the following, as well as in the simulation of \figurename~\ref{fig:coverage}.
\end{remark}

\begin{definition}[$(\tau, K)$-coverage regulation]
Let $\Psi$ be a stationary PP. Let $K$ and $\tau$ be positive real numbers. We say that $\Psi$ is $(\tau, K)$-covering if, $\forall n \in \N, \quad \Psi(b(o,n\tau))\geq \lceil K n^2 \rceil$, $\Pr$~-~almost-surely.
\end{definition}

\begin{lemma}\label{lem:equivalence1}
    A stationary PP $\Psi$ is $\tau$-void regulated if and only if it is $(\tau,\frac{\pi}{6\sqrt{3}})$-covering.
\end{lemma}

\begin{IEEEproof} See Appendix~\ref{proof:equivalence1}.\end{IEEEproof}

\begin{proposition}\label{prop:cumulativeCF} Let $\Psi$ be a $(\tau, K)$-covering PP, and  $\ell(r) = \Theta(r^{-\beta}),\ \beta > 2$. For any transmitter location $x \in \mathbb{R}^2$, 
\begin{equation}
    \sum_{y\in\Psi} \ell(\|x-y\|) = \Omega\left( K\tau^{-\beta} \right), \quad \Pr-a.s.
\end{equation}
\end{proposition}

\begin{IEEEproof} See Appendix~\ref{proof:cumulativeCF}.\end{IEEEproof}

Void regulation is used in Corollary~\ref{prop:corollary} in order to provide guarantees for the whole network.  \figurename~\ref{fig:coverage} illustrates how a spatially regulated PP of transmitters can be jointly constructed with a void regulated PP of receivers. Coverage regulation is used to provide deterministic guarantees for CF networks in Section~\ref{sec:CF}.

\subsection{Application to \texorpdfstring{\SIC}{SIC}} \label{sec:SIC}

As explained above, the \SIC~framework is as follows: each message is decoded by a receiver $y\in\Psi$, only if the messages of transmitters with higher received power w.r.t.  $y$ are successfully decoded by $y$.  The \SINR~(see \eqref{def:SRINR}) \textit{coverage threshold} $\eta_0$ is the level of SINR that must be satisfied in order for a message to be decoded with \SIC: if the \SINR~is not above that threshold, the receiver fails 
decoding the transmission. The fact that a receiver is able to decode messages from a transmitter or 
not is completely determined by its distance to the transmitter, the fading statistics (when fading is taken into account, see Section~\ref{sec:fading}), and by the level of residual interference. We assume perfect interference cancellation, 
i.e., when a message is successfully transmitted, its signal is completely cancelled from the received signal. 
All signals with \SINR~lower than $\eta_0$ are considered lost. 
This is well studied in \cite{zhang2014performance} 
in the case of PPP. The specificity of \SIC, is that for each transmission, the set of messages that 
must be decoded before a given message is finite, which in our case provides control on latency. 
Moreover, we study the impact of spatial regulation that allows one to upper-bound interference $\Pro$-almost-surely.
Let $\sum_{x\in\Phi\setminus B(y,r)} \ell(\|x-y\|) $ denote the interference at the origin due to transmitters at a distance more than $r>0$ to the receiver $y\in\Psi$, and with an upper-bounded, non-increasing, non-negative path-loss $\ell(r) = \Theta(r^{-\beta}), \beta > 2$. 
Note that, in this definition, we substract the closed ball $B(o,r)$ from $\Phi$ because, in \SIC, in the no
fading case, the interferers of a transmission at distance $r$ are the transmitters with a larger distance. 
From SNC, the interference in \SIC~can be bounded deterministically from above, i.e., $\Pro$-almost-surely. 
This is illustrated by the following proposition, which is an extension of 
\cite[Corollary~3]{feng2023spatial} to the case of $(\sigma, \rho, \nu)$-ring shot-noise regulation.

    \begin{figure}[t]
        \centering
        \begin{tikzpicture}
    \begin{groupplot}[
        width=\linewidth,
        height=4.5cm,
        grid=major,
        xlabel={Distance (m)},
        title style={yshift=0.5ex},
        legend columns=3,
        legend style={
            at={(0.5, 0.75)}, 
            anchor=south, 
            yshift=1.0cm,          
            draw=black,            
            fill=white,            
            thin                   
        }
    ]
    
    \pgfplotsinvokeforeach{8}{
        \nextgroupplot[
            ylabel={\ifnum#1=8 $\mathrm{SINR}$ (dB)\fi}
        ]
        
        \addplot[color=red, dashed, mark=*, smooth] 
            table [x=Distance_m, y=HPP_H#1, col sep=comma] {figures/sinr_HPP_PPP.csv};
        \ifnum#1=8 \addlegendentry{$\SINR$ MHPP} \fi
    
        \addplot[color=blue, solid, mark=triangle*, smooth] 
            table [x=Distance_m, y=Bound_H#1, col sep=comma] {figures/sinr_HPP_PPP.csv};
        \ifnum#1=8 \addlegendentry{Eq. \eqref{eq:SINRbound} bound} \fi
    
        \addplot[color=black, dashed, mark=square*, smooth] 
            table [x=Distance_m, y=PPP_H#1, col sep=comma] {figures/sinr_HPP_PPP.csv};
        \ifnum#1=8 \addlegendentry{$\SINR$ PPP} \fi
    }
    \end{groupplot}
\end{tikzpicture}
        \caption{\SINR~w.r.t.  distance, with $\ell(r) = \max\{1, r\}^{-4}$ and $\gamma_0=-10$ dB. The simulated PP is a MHPP  of type II, where points from a stationary PPP of intensity $\lambda=1$ are kept only if they are  at least at distance $2H=8$ from each other. Hence the intensity of the Matérn process is $\lambda_H = \frac{1-e^{-4\lambda H^2\pi}}{4H^2\pi}$, c.f. \cite[p. 58]{haenggi2013stochastic}. In order to compare the \SINR~coming from Matérn (red circles) and Poisson (black squares) PPs, the PPPs are with intensity $\lambda_H$. The black and red curves correspond to minimal \SINR~values encountered in simulations w.r.t.  to distance, while the blue curve corresponds to the bound provided in \eqref{eq:interferencebound}. The \SINR~is computed for 1\,000 simulated PPs.}
        \label{fig:rate}
    \end{figure}
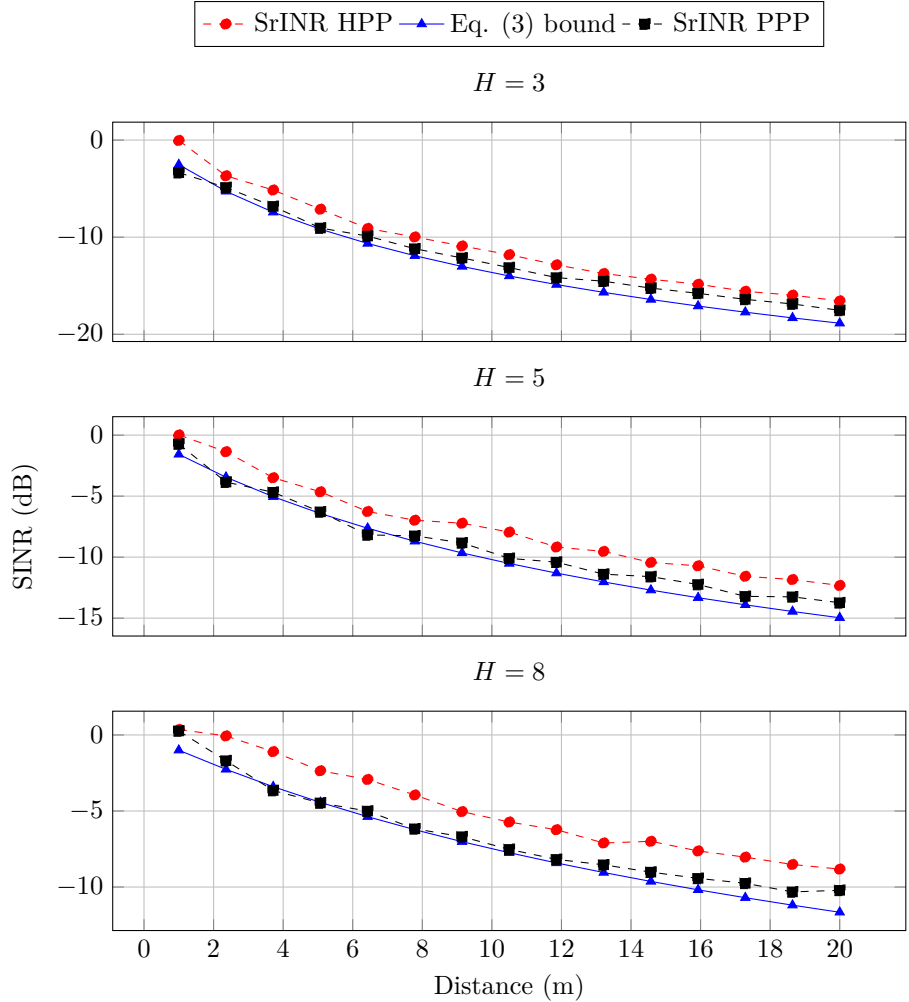

\begin{proposition}\label{prop:shotnoisenofading}
Under the foregoing assumptions, for all PPs $\Phi$ which are $(\sigma, \rho, \nu)$-ring regulated w.r.t. $\Psi$, and $\ell(r) = \Theta(r^{-\beta}), \ \beta > 2$, 
    \begin{equation}\label{eq:interferencebound}
        \sum_{x\in\Phi \setminus B(o,r)} \!\!\!\!\ell(\|x\|) = O\left(\sigma r^{-\beta}   + \frac{\rho}{\beta-1} r^{1-\beta} + \frac{2\nu}{\beta-2} r^{2-\beta}\right),
    \end{equation}
   $\Pro$-almost-surely.
\end{proposition}

\begin{IEEEproof}
Since $\Phi$ is $(\sigma, \rho, \nu)$-ring regulated, it is shot-noise regulated according to Lemma~\ref{lem:equivalence}. Since $\ell(r) = \Theta(r^{-\beta})$, let $ C  > 0$ be constants such that $\ell(r) \leq C r^{-\beta}$ for all $r > 0$. Thus, $\Pro$-almost-surely,
    \begin{multline*}
        \sum_{x\in\Phi\setminus B(o,r)} \ell(\|x\|) = \lim_{R \to \infty} \sum_{y\in\Phi\cap\o(r, R)} \ell(\|y\|)\\
        \leq \lim_{R \to \infty} \sigma\ell(r) + \rho \int_{r}^{R} \ell(s)ds + 2\nu \int_{r}^{R}s\ell(s)ds, \\
         \leq \lim_{R \to \infty} \sigma C r^{-\beta} + \rho C \int_{r}^{R} s^{-\beta}ds + 2\nu C\int_{r}^{R}s^{1-\beta}ds, 
    \end{multline*}
    and we conclude by integrating which is well defined since we require $\beta > 2$.
\end{IEEEproof}

\begin{remark}\label{rmk:HPPexample}
    From Lemma~\ref{lem:equivalence2}, Remark~\ref{rmk:hardcore} and Proposition~\ref{prop:shotnoisenofading}, if $\ell(r) = r^{-\beta}, \ \beta > 2$, a MHPP of hardcore distance $H > 0$ is such that $$\sum_{x\in\Phi\setminus B(o,r)} \ell(\|x\|) \leq  \frac{\pi^2}{\sqrt{3}} \left(  r^{-\beta}   + \frac{r^{1-\beta}}{H(\beta-1)}  + \frac{r^{2-\beta}}{2H^2(\beta-2)} \right),$$ $\Pro$-almost-surely.
\end{remark}

\begin{remark}
    The upper-bound provided in Proposition~\ref{prop:shotnoisenofading} holds with probability $1$. Note that in the case $\ell(r) = Cr^{-\beta}, \beta > 2$, and Matérn hard-core processes of intensity $\lambda$ and hard-core distance $H$, it is shown  that the average interference above a distance $R > 0$ is $\Eo\left[\sum_{x \in \Phi\setminus b(o,R)} \ell(\|x\|)\right] = \frac{2\pi\lambda C\exp \left\{ -\pi\lambda H^2 \right\}}{R^{\beta - 2}(\beta - 2)}.$ See \cite[8.9.2]{haenggi2013stochastic}. Note that the bound in \eqref{eq:interferencebound} does not depend on $\lambda$.
\end{remark}

\begin{figure}[t]
        \centering
    \subfloat[$\Psi$ $\tau_0$-void and $\Phi$ ring regulated]{
       \includegraphics[width=0.7\linewidth]{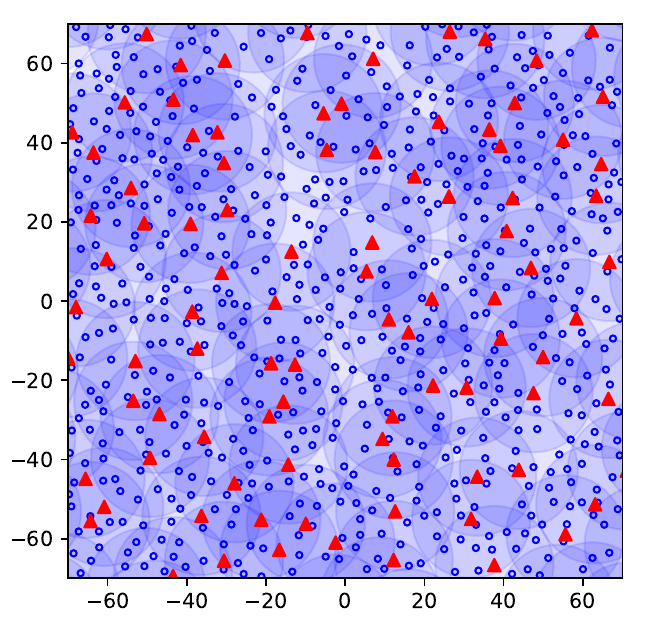}
     \label{fig:voidandring}
    }
    \\
    \subfloat[$\Phi$ ring regulated, $\Psi$ not void regulated.]{
        \includegraphics[width=0.7\linewidth]{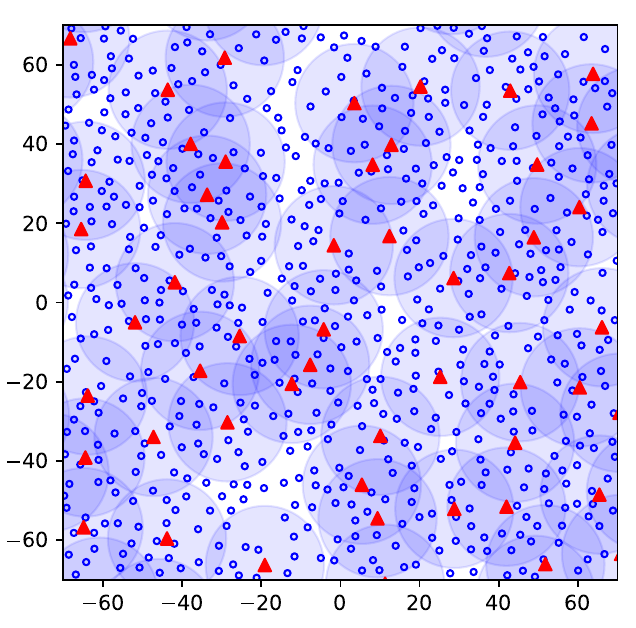}%
    \label{fig:justring}
    }    

   \caption{The space regulated network with parameters shown in Table~\ref{tab:simulations}. The blue dots are transmitters, the red triangles receivers. Transmitters are separated by at most distance $2H=6$. Receivers are separated by at most $\tau_0$ units of space, here $\eta_0 = -10$dB. The blue disks are of radius $\tau_0.$ One can see that in Figure (a), the whole Euclidean plane is covered, i.e., all the transmitters have at least a receiver that is able to decode their messages in a decoding delay lower than $T$. In Figure (b), there are some blank areas from where transmissions will never be decoded.}
    \label{fig:coverage}
\end{figure}

The \textit{signal-to-residual-interference-plus-noise ratio} (SrINR) of the receiver $y\in\Psi$ w.r.t.  the transmitter $x\in\Phi$ is hence $\frac{\ell(\|x-y\|)}{\sum_{w\in\Phi\setminus B(y,\|x-y\|)} \ell(\|w-y\|) + N_0/P_0}$ under \SIC, which is identically distributed  to $\SINR(x-y,\theta_y\Phi)$ where \begin{equation}\label{def:SRINR}\SINR(x,\Phi) = \frac{\ell(\|x\|)}{\sum_{w\in\Phi\setminus B(o,\|x\|)} \ell(\|w\|) + \gamma_0}.\end{equation}
Since $\Phi$ is $(\sigma, \rho, \nu)$-shot-noise regulated w.r.t.  $\Psi$, and $\ell(r) = \Theta(r^{-\beta)}$, from \eqref{eq:interferencebound} we get that, for $r = \|x\|$, $\SINR(x,\Phi) = \Omega(\eta(r))$,  $\Pro$-almost surely with \begin{equation}\label{eq:SINRbound}
    \eta(r) \triangleq \left( \sigma  + \frac{\rho}{\beta-1} r + \frac{2\nu}{\beta-2} r^2 + \gamma_0 r^\beta \right)^{-1}.
\end{equation}

\begin{remark} \label{rmk:HPPexample1}
    From Remark~\ref{rmk:HPPexample}, if $\ell(r) = r^{-\beta}, \ \beta > 2$, a MHPP of hardcore distance $H > 0$ is such that $$\eta(r) \geq   \left( \frac{\pi^2}{\sqrt{3}} + \frac{\pi^2}{H\sqrt{3}(\beta-1)} r + \frac{\nu}{H^2\sqrt{12}(\beta-2)} r^{2} + \gamma_0 r^\beta\right)^{-1},$$ $\Pro$-almost-surely.
\end{remark}

In order to provide deterministic guarantees, we analyze the conditions under which the \SINR~is $\Pro$-almost-surely larger than the threshold $\eta_0$.

\begin{definition}[Coverage distance]Let $y\in\Psi$ and \begin{equation}\label{def:coveragedistance}
    \tau(\Phi) \triangleq \inf\left\{\|x\| : x\in\Phi, \ \eta_0 > \SINR(x,\Phi) \right\}
\end{equation} be the coverage distance of the typical receiver for decoding with \SIC, such that $\tau(\theta_y\Phi)$ is the smallest distance at which receiver $y\in\Psi$ fails to decode transmitters under \SIC. \end{definition}

\begin{remark}All deterministic guarantees for \SIC~can only be given for transmitters at distance $\tau_0$ at most from a receiver. We apply this property in the following by studying two cases : the case of a transmitter PP which is ring regulated w.r.t. the receivers PP (see Proposition~\ref{prop:uplinkdelay}), and the case where transmitters are ring regulated w.r.t. the receivers and the receivers are void regulated w.r.t. transmitters (see Corollary~\ref{prop:corollary}).\end{remark}

\begin{figure}[t]
     \centering
     \begin{tikzpicture}
\begin{axis}[
    width=\linewidth, 
    height=5cm,
    xlabel={$\eta_0$ (dB)},
    ylabel={$\tau_0$ (m)},
    grid=major,
    grid style={dashed, gray!30},
    legend pos=north east,
    legend style={nodes={scale=0.9, transform shape}},
    thick,
    xmin=-30, xmax=0 
]

\addplot[color=red, line width=1.5pt, smooth] 
    table [x=eta_0_dB, y=tau_gamma_50dB, col sep=comma] {figures/tau_eta_data.csv};
\addlegendentry{$\gamma_0 = -50$ dB}

\addplot[color=green!60!black, line width=1.5pt, smooth] 
    table [x=eta_0_dB, y=tau_gamma_20dB, col sep=comma] {figures/tau_eta_data.csv};
\addlegendentry{$\gamma_0 = -20$ dB}

\addplot[color=blue, line width=1.5pt, smooth] 
    table [x=eta_0_dB, y=tau_gamma_10dB, col sep=comma] {figures/tau_eta_data.csv};
\addlegendentry{$\gamma_0 = -10$ dB}

\end{axis}
\end{tikzpicture}
   \caption{Coverage distance lower-bound w.r.t.  \SINR~threshold for $\ell(r) = \max\{1, r\}^{-4}$, $\alpha=1/2$.}
    \label{fig:tau}
    \end{figure}
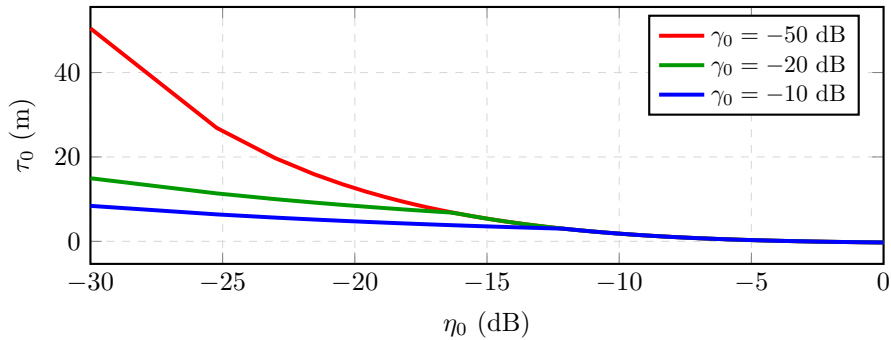

\begin{proposition}\label{prop:coverage}
    Let $(\Phi, \Psi)$ be two jointly stationary PPs with $\Phi$ 
    $(\sigma, \rho, \nu)
    $-ring regulated w.r.t.  $\Psi$. Let $\ell(r) = \Theta(r^{-\beta}), \ \beta > 2$, be the path-loss of the point-to-point channels. Let $\tau(\Phi)$ be the coverage distance of the typical receiver as defined in \eqref{def:coveragedistance}. If $\exists \alpha \in (0,1)$ such that $\eta_0 > \frac{\sigma}{\alpha} - \frac{\rho^2 (\beta-2)}{8 \nu \alpha (\beta-1)^2}$, then there exists $ C > 0$ such that, $\tau(\Phi) \geq \tau_0, \  \Pro-\text{almost-surely}$, where $\tau_0 \triangleq \min \left\{ 
    \sqrt{ \frac{\rho^2 (\beta-2)^2}{16\nu^2(\beta-1)^2} + \frac{\beta-2}{2\nu} \left( \frac{\alpha}{C \eta_0} - \sigma \right) } - \frac{\rho(\beta-2)}{4\nu(\beta-1)}, 
    \left( \frac{1-\alpha}{C \eta_0 \gamma_0} \right)^{1/\beta} 
\right\}$.
\end{proposition}

\begin{IEEEproof}
    See Appendix~\ref{proof:coverage}.
\end{IEEEproof}

\begin{remark}
    Note that it may happen that, for some values of the parameters, the first inequality in \eqref{eq:SIRbound} has no solution. Hence, there are two cases : either $\eta_0 > \frac{\sigma}{\alpha} - \frac{\rho^2 (\beta-2)}{8 \nu \alpha (\beta-1)^2}$ and the SIR bound impacts the coverage. Otherwise the second equality in Proposition~\ref{prop:coverage} provides $\tau_0$.
\end{remark}

This bound is not necessarily tight in the sense that since $\tau(\Phi) \geq \tau_0$, there could be some transmitters further away than $\tau_0$ from any receiver which can nevertheless be decoded. Furthermore, we see that, when $\eta_0$ goes to 0, $\tau_0$ goes to infinity, which shows that the whole set of
transmitters can be decoded for codes based on low enough rate, under void regulation assumptions. It is also clear from \figurename~\ref{fig:tau} that $\tau_0$ decreases fast w.r.t. $\eta_0$. It should be clear that one can guarantee that
the transmitters that are at distance less than $\tau$ from a receiver
can all be decoded. 
We see that a low value of
the threshold $\eta_0$ permits the receivers to cover all transmitters. 
We also observe in \figurename~\ref{fig:tau} that $\tau_0$ quickly gets small when $\eta_0$ increases. 
Thus, we are deterministically assured that every transmission is covered with \SIC~using $\tau_0$-void regulation. Void regulation allows one to guarantee that at least one receiver is at distance $\tau_0$ from any transmitter. One can use this property to guarantee that either the network is

\begin{enumerate} 
    \item  both space regulated w.r.t. receivers and void-regulated w.r.t. transmitters, see \figurename~\ref{fig:voidandring}. In this case, the decoding delay guarantee is provided for the whole network.
    \item only spatially regulated w.r.t. receivers, see \figurename~\ref{fig:justring}. In this case, only a subset of transmitters, namely 
    \begin{equation}\label{eq:initPP}
        \Phi_0 \triangleq \bigcup_{y\in\Psi} \Phi \cap b(y, \tau_0),
    \end{equation} 
    are deterministically guaranteed both in reliability (\SINR~above $\eta_0$) and (bounded) time. In that case, a decoding delay guarantee is provided for transmitters at distance $\tau_0$ at most.
\end{enumerate}

\section{Delay analysis for cellular networks} \label{sec:time}

The objective in this section is to determine the minimal value of the span $T$ such that the transmitted message from $x\in\Phi$ can be decoded almost surely, provided that $x$ is in the network coverage area of the typical receiver. In order to provide decoding delay guarantees, we need to provide deterministic decoding delay bounds, i.e., to find $T$ such that \begin{equation}\label{eq:decodingdelay_req}\Pr(\Pro(D(x,\Phi) > T \mid \Phi) \leq \epsilon)=1,\end{equation} for the \SIC~decoding delay $D(x,\Phi)$ to be defined later in \eqref{eq:delay_bound_palm}. Wireless networks are inherently stochastic. To provide strict guarantees on decoding delays themselves is impossible in general. Nevertheless, \eqref{eq:decodingdelay_req} is to be understood as a deterministic guarantee, as it provides a deterministic threshold on the probability that all transmissions meet their delay requirement, for any realization of the locations of the transmitters and receivers, as long as they satisfy the spatial regulation properties introduced in the previous section. In the following, we use this framework to provide such guarantees using the rates provided by spatial regulation, firstly in Section~\ref{sec:sicdecoding delays}, in the case without fading, where $\epsilon = 0$, and in Section~\ref{sec:fading}, the case with fading, where $\epsilon > 0$. At this stage, we assume that transmitters are fully synchronized to send their messages at the beginning of the time slots. See \figurename~\ref{fig:sic}. Hence, in each time interval of the form $(kT, (k+1)T], k \geq 1$, all transmitted messages are decoded iteratively by all receivers using \SIC. We say that the decoding delay guarantee is fulfilled from the transmitter $x\in\Phi_0$ if all receivers $y\in\Psi$ such that $\SINR(x-y, \theta_y\Phi) > \eta_0$ decode the message by the end of the time slot with \SIC. With this model, the performance analysis reduces to the analysis on a typical time slot, say $[0,T]$. Thus, the purpose of this section is to find the minimal span $T$ such that \eqref{eq:decodingdelay_req} holds for the \SIC-decoding delays defined in Section~\ref{sec:sicdecoding delays}. In order to do so, we bound from below what the network is able to provide in terms of decoding power, using the spatial regulation discussed in the previous section.

The goal of this section is to find $T$ such that all transmissions coming from $\Phi_0$ can be decoded in one time slot of size $T$.


\subsection{Decoding delay}

We assume that the Shannon capacity is achieved by the decoder in the point-to-point channels.  Hence the ergodic rate of a link from the transmitter $x\in\Phi$ to the typical receiver located at the origin using \SIC~is \begin{equation}\label{def:rate}R(x, \Phi) = \inf_{w\in\Phi\cap b(o,\|x\|)}\log_2\left(1 + \SINR(w, \Phi)\right),\end{equation} 
for the link from $x\in\Phi$ to the typical receiver located at the origin.

\begin{lemma}\label{lem:constantrateSIC}
    Let $(\Phi, \Psi)$ be two jointly stationary PPs with $\Phi$ 
    $(\sigma, \rho, \nu)
    $-ring regulated w.r.t.  $\Psi$. Let $x\in\Phi$ and $y\in\Psi$. The rate $R(x-y,\theta_y\Phi)$ defined in \eqref{def:rate} is almost-surely lower-bounded by $\log_2\left(1+\eta(\|x-y\|)\right)$, where $\eta$ is defined in \eqref{eq:SINRbound}.
\end{lemma}

\begin{IEEEproof}
    Firstly, for $r = \|x-y\|$, we can write the \SINR~as $$\SINR(x-y,\theta_y\Phi) = \left(\frac{\sum_{w \in\Phi\setminus B(y,r)} \ell(\|w-y\|)}{\ell(r)} + \frac{\gamma_0}{ \ell(r)}\right)^{-1}.$$ Let us now apply  Proposition~\ref{prop:shotnoisenofading}, \begin{equation}
        \SINR(x-y, \theta_y\Phi) \geq \eta(\|x-y\|), \quad \Pry-a.s.,
    \end{equation} and since $\eta(r)$ is decreasing, we get $\inf_{w:\|w-y\| \leq \|x-y\|}R(w-y,\theta_y\Phi) \geq \log_2(1 + \eta(\|x-y\|)), \Pry$-almost-surely. We conclude with the joint stationarity of $(\Phi, \Psi)$.
\end{IEEEproof}


\subsection{Application to \texorpdfstring{\SIC}{SIC} without fading}\label{sec:sicdecoding delays}
 The \textit{virtual decoding delay} $d_{xy}$ is the number of channel uses required by the receiver $y\in\Psi$ to decode a message from $x$, assuming all stronger received power messages have been decoded. We suppose that there are no decoding error, i.e., that the message sent by $x \in \Phi$ is decoded by $y\in\Psi$ with no error, as long as $\SINR(x-y,\theta_y\Phi) > \eta_0$. 
In this section, we apply the foregoing to all links.

Let us consider the \textit{worst-case decoding delay} $D(x,\Phi)$, when all receivers decode with \SIC. We define this decoding delay properly in the following, and upper-bound it using the virtual decoding delays from the NC framework. 

\begin{definition}[Virtual decoding delay]\label{def:sicdecoding delays} The virtual decoding delay of the link from the transmitter $x\in\Phi$ to the typical receiver is
    \begin{equation}\label{eq:sicdecoding delays}
        d(x,\Phi) \triangleq \frac{m}{R(x,\Phi)},
    \end{equation} which represents the decoding delay from $x$ to the typical receiver supposing that no other transmission with higher received power than $x$ is transmitted simultaneously.
\end{definition}

Let us consider a simple example, where a receiver $y\in\Psi$ decodes two transmissions, one from $x$ and one from $w\in b(y,\|x-y\|)$. Receiver $y$ receives power from $w$ than $x$. Hence, the message from $x$ can be theoretically \jmg{decoded} after $w$ by receiver $y$ at a rate equal to $R(x-y, \theta_y\Phi)~ \text{bpcu}$.  Either the message of $x$ is decoded before $w$, or not, which leads to two cases. The \SIC-decoding delay from $x\in\Phi$ to $y\in\Psi$ is $\Pry$-almost surely upper-bounded either by $d(x-y,\theta_y\Phi)$  (the virtual decoding delay supposing the message from $x$ has been decoded by $y$ without the message from $w$) or the uplink decoding delay of $w$ in \SIC, denoted by $d(w-y,\theta_y\Phi)$ otherwise, because, once $w$ is decoded, $x$ can be immediately decoded, as explained in Section~\ref{sec:sysmodel}. 
Thus the general formula is obtained by induction: we get that the \textit{decoding delay} of a transmission from $x$ to receiver $y$ is upper-bounded by $\sup_{w\in \Phi \cap b(y, \|x-y\|)} d(w-y,\theta_y\Phi)$. 
Thus, we define the decoding delay accounting \fb{for} the latter.

\begin{definition}[\SIC-decoding delay]\label{def:SICdelay}
    The \SIC-decoding delay $D(x,\Phi)$ of a transmission from $x\in\Phi$ to the typical receiver is the largest virtual decoding delay among the set of transmissions with higher received power than $x$ w.r.t.  the typical receiver, i.e.,\begin{equation}\label{eq:delay_bound_palm}D(x,\Phi) \triangleq \sup_{w\in\Phi\cap b(o,\|x\|)}d(w,\Phi).\end{equation} 
\end{definition}






We show in the following that with ring regulation, $D(x,\Phi)$ can be upper-bounded almost-surely.

    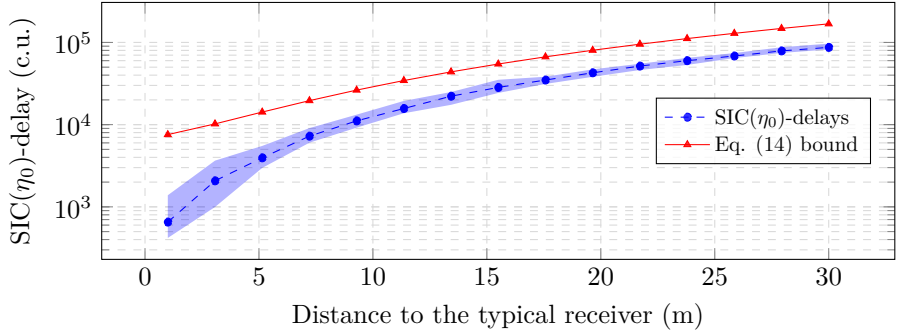
\begin{figure}[t]
        \centering
        \begin{tikzpicture}
\begin{axis}[
    width=\linewidth,
    height=5cm,
    ymode=log, 
    xlabel={Distance to the typical receiver (m)},
    ylabel={\SIC-delay (c.u.)},
    legend style={
        at={(0.97,0.5)},
        anchor=east,
        nodes={scale=0.8, transform shape},
        fill opacity=0.8,
        draw opacity=1,
        text opacity=1
    },
    legend cell align={left},
    grid=both,
    grid style={dashed, gray!30}
]

\addplot[draw=none, name path=min_del, forget plot] table [x=r, y=min_delay, col sep=comma] {figures/delay_data.csv};
\addplot[draw=none, name path=max_del, forget plot] table [x=r, y=max_delay, col sep=comma] {figures/delay_data.csv};

\addplot[blue, fill opacity=0.3, forget plot] fill between[of=min_del and max_del];

\addplot[blue, dashed, mark=*, mark size=1.5pt] table [x=r, y=mean_delay, col sep=comma] {figures/delay_data.csv};
\addlegendentry{\SIC-delays}

\addplot[red, solid, mark=triangle*, mark size=1.5pt] table [x=r, y=bound_delay, col sep=comma] {figures/delay_data.csv};
\addlegendentry{Eq. \eqref{eq:delaySRbound} bound}

\end{axis}
\end{tikzpicture}%
        \caption{Decoding delays simulations and theoretical bound against distance to receivers, with the parameters shown in Table~\ref{tab:simulations}.}
        \label{fig:decoding delay}
    \end{figure}

The message size being constant for all transmitters, the challenge lies in characterizing the distribution of the decoding delays. 
Lemma~\ref{lem:constantrateSIC} tells us that, for a receiver $y\in\Psi$ and a transmitter $x \in \Phi \cap b(y,\tau_0)$, the decoding rate of the transmission from $x$ to $y$ is lower-bounded by $\log_2\left(1+\eta(\|x-y\|)\right)$ in the synchronous case. Hence the decoding delay of a transmission from $x\in\Phi$ to a receiver $y$ is $\Pr$-almost surely upper-bounded by $\frac{m}{\log_2(1+\eta(\|x-y\|))}$, and $d(x, \Phi) \leq \frac{m}{\log_2(1+\eta(\|x\|))},  \Pro-a.s.,$ where $\frac{m}{\log_2(1+\eta(r))}$ is shown in \figurename~\ref{fig:decoding delay} as a function of $r>0$. 





\begin{proposition}\label{prop:uplinkdelay}
          Let $(\Phi, \Psi)$ be two jointly stationary PPs modeling a cellular network, where $\Psi$ decodes messages from $\Phi$ using \SIC, and such that $\Phi$ is $(\sigma, \rho, \nu)$-ring regulated w.r.t. $\Psi$. Suppose there is an upper-bounded path-loss $\ell(r) = \Theta(r^{-\beta}), \ \beta > 2$, between transmitter and receivers. Suppose all transmitters send messages of size $m$.  Let $D(x,\Phi)$ be the \SIC-decoding delay of a transmission from $x\in\Phi$ to the typical receiver defined in \eqref{eq:delay_bound_palm}. Then 
     \begin{equation}\label{eq:delaySRbound}
         D(x,\Phi) \leq \frac{m}{\log_2\left(1+\eta(\|x\|)\right)}, \quad \Pro-a.s.,
    \end{equation} where $\eta$ is given in \eqref{eq:SINRbound}. Furthermore,  if $\exists \alpha \in (0,1)$ such that $\eta_0 > \frac{\sigma}{\alpha} - \frac{\rho^2 (\beta-2)}{8 \nu \alpha (\beta-1)^2}$, then $\sup_{x\in\Phi_0} D(x,\Phi)  \leq \frac{m}{\log_2\left(1+\eta_0\right)}$, $\Pro$-almost-surely.
\end{proposition}

\begin{IEEEproof} 
     See Appendix~\ref{appendix:sicdecoding delay}.
\end{IEEEproof}

\begin{corollary}\label{prop:corollary}
     Same assumptions as Proposition~\ref{prop:uplinkdelay}. Assume in addition that $\exists \alpha \in (0,1)$ such that $\eta_0 > \frac{\sigma}{\alpha} - \frac{\rho^2 (\beta-2)}{8 \nu \alpha (\beta-1)^2}$, and that $\Psi$ is $\tau_0$-void regulated w.r.t. $\Phi$ where $\tau_0$ is defined in Proposition~\ref{prop:coverage}. Then,  \begin{equation*}
        \sup_{x\in\Phi}\inf_{y\in\Psi} D(x-y,\theta_y\Phi)  \leq \frac{m}{\log_2\left(1+\eta_0\right)}, \quad \Pr-a.s.
    \end{equation*}
\end{corollary}

\begin{IEEEproof}
    Thanks to Proposition~\ref{prop:coverage}, $\tau_0$ exists with $\eta_0 >\frac{\sigma_\beta}{\alpha} - \frac{\rho_\beta^2}{4\nu_\beta\alpha}$. Since $\Psi$ is $\tau_0$-void regulated w.r.t. $\Phi$, $\Pr$-almost surely all transmitters $x\in\Phi$ and receivers $y\in\Psi$ are such that $\|x-y\| \leq \tau_0$, thus with Proposition~\ref{prop:uplinkdelay}, $\Pr(\inf_y D(x-y,\theta_y\Phi)  > \frac{m}{\log_2\left(1+\eta_0\right)} \mid \Phi, \Psi) \leq \prod_{y\in\Psi} \Pro( D(x-y, \theta_y\Phi)  > \frac{m}{\log_2\left(1+\eta_0\right)} \mid \Phi) = 0$ since there exists at least one receiver for each transmitter such that $\|x-y\| \leq \tau_0$. 
\end{IEEEproof}

This last section describes the delays in a cellular network. Void-regulation ensures that every transmitter can be decoded in its own cell of radius $\tau_0$. Proposition~\ref{prop:uplinkdelay} provides bounds on the delay such that the whole cellular network is covered, i.e., every transmitter is decoded by at least one receiver. In the next section, we consider the CF network and its diversity benefits, using coverage regulation. Before this, we provide the core method that ensures that delay guarantees in cellular networks hold in the fading case.

\subsection{Application to \texorpdfstring{\SIC}{SIC}~with i.i.d. fading} \label{sec:fading}

    In the following, we provide deterministic guarantees on the probability the \SIC-decoding delay $D(x,\Phi)$ is almost-surely smaller than the span $T$, in the absence of fading. Accounting for fading should lead to soften the constraints on decoding delays in the foregoing, i.e. upper-bound the tail distribution of delays such as $\Pro(D(x,\Phi) > T \mid \Phi ) \leq \epsilon,$ $\Pro$-almost-surely. We have done so with $\epsilon =0$ in the absence of fading. In \cite{feng2023spatial} such application is provided without \SIC. Note that the CF aspect of the network is not considered in this section. 
    
    The received signal power from $x\in\Phi$ to the typical receiver is denoted by $P_x = P_0 h_x \ell(\|x\|)$, where $h_x$ is i.i.d. fading w.r.t. the channel between $x\in\Phi$ and the typical receiver. We suppose that the channel remains unchanged during one time slot. Hence the channel fading distribution is assumed fixed and denoted by $\mu$. Let us denote by $I_x = \sum_{w\in\Phi: P_w < P_x} P_w$  the sum of lower received power than $x\in\Phi$, namely the interference on $x\in\Phi$, we get the \SINR~expression
    \begin{equation}\label{eq:SINRwithFading}
        \SINR(x, \Phi) = \frac{P_x}{I_x + N_0}.
    \end{equation}
    Since we suppose that for the duration of a time slot, the channel remains constant, hence $h_x$ is sampled once for each time slot. All we need to ensure is that the Laplace transform of the rate $R(x,\Phi)$ is upper-bounded (see \eqref{eq:defRLregulation}), which requires additional conditions on the fading distribution $\mu$, and more specifically its moment generating function (m.g.f.). For the purpose of proving closed expressions, we fix the path loss function as $\ell(r) = \max\{1, Cr\}^{-\beta}, \beta >2$. 
    The difficulty to apply ring regulation to \SIC~lies in the fact that received power is not only dependent on the distance to receivers but also from (random) fading. In this section, we show in Proposition~\ref{prop:rateSpaceregulated}, that, under some conditions on the Laplace transform of $\mu$, there exists $\zeta(\theta ; p) > 0$ such that 
    \begin{align}\label{eq:defRLregulation}
        \Eo[\exp\left\{-\theta R(x,\Phi)\right\} \mid \Phi, P_x] \leq \exp\{-\zeta(\theta; P_x)\},
    \end{align}
    $\Pro$-a.s., which is sufficient, in addition to the foregoing section, to provide delay guarantees.


    \begin{lemma}\label{lem:SNCvirtualdelay}
           If \eqref{eq:defRLregulation} holds, the virtual decoding delay is such that $$\Pro(d(x, \Phi) > T \mid \Phi, P_x) \leq \inf_{\theta > 0}\frac{\exp\{m\theta - T \zeta(\theta; P_x))\}}{1 - \exp\{- \zeta(\theta; P_x)\}}, $$ $\Pro$-almost-surely.
    \end{lemma}

    \begin{IEEEproof} 
        See Appendix~\ref{proof:SNCvirtualdelay}.
    \end{IEEEproof}
    
 %
 %
    
    Let us first bound the interference term with the following lemma. 

    \begin{proposition}\label{prop:fadingSpaceregulated}
        Let $\Phi$ be a $(\sigma, \rho, \nu)$-ring regulated PP w.r.t.  to $\Psi$. Consider i.i.d. fading with distribution $\mu$, Laplace transform $L(s) = \int_0^\infty e^{sh}\mu(dh)$ and $k$-th order moment $m_k = \int_0^\infty h^k\mu(dh)$, and a path loss function $\ell(r) =\max\{1, Cr\}^{-\beta}$. If $\mu$ admits exponential moments, i.e., is such that $s_\star = \sup\{s \in \R_+:  \forall \vert z \vert  < s, L(zP_0) < \infty\}$ exists, then $\Pr$-almost-surely, $\forall s \in [0, s_\star]$,
             $\Eo \left[ \exp\left\{s I_x \right\} \mid \Phi, P_x\right] \leq \exp\left\{A(sP_0) + B\left(\frac{P_x}{P_0}\right)\right\}$,  where $A(z) = \left(\sigma + \frac{\rho}{C} +\frac{\nu}{C^2}\right) (L(z)-1) + \frac{\rho}{C\beta} \sum_{k=1}^\infty \frac{z^km_k }{k!(k - 1/\beta)} + \frac{\nu}{C^2\beta} \sum_{k=1}^\infty \frac{z^km_k}{k!(k - 2/\beta)}$ and $B(p) = \sigma \mu\left(h \in \R_+ : h > p\right) +  \frac{\rho }{p^{1/\beta}} \int_{p}^\infty h^{1/\beta}\mu(dh)+ \frac{\nu }{p^{2/\beta}} \int_{p}^\infty h^{2/\beta}\mu(dh).$
    \end{proposition}

    \begin{IEEEproof}
        See Appendix~\ref{proof:fadingSpaceregulated}.
    \end{IEEEproof}

   \begin{corollary}\label{coro:rayleigh}
        Same assumptions as Proposition~\ref{prop:fadingSpaceregulated}. With Rayleigh fading, $s_\star=1/P_0$, and $\forall s \in (0, 1)$, $A(s) =  \frac{s(\sigma + \frac{\rho}{C} +\frac{\nu}{C^2})}{1 - s} + \frac{\rho}{C} \frac{s}{\beta - 1} {}_2F_1\left(1, \frac{\beta-1}{\beta}, \frac{2\beta-1}{\beta}; s\right) + \frac{\nu}{C^2} \frac{s}{\beta - 2} {}_2F_1\left(1, \frac{\beta-2}{\beta}, \frac{2\beta-2}{\beta}; s\right)$, and $\forall p > 0$, $B(p) = \sigma e^{-p} +  \frac{\rho}{p^{1/\beta}} \Gamma\left(\frac{\beta+1}{\beta}, p\right) +\frac{\nu}{p^{2/\beta}} \Gamma\left(\frac{\beta+2}{\beta}, p\right),$ 
        where $\Gamma(a,z)$ is the upper-incomplete Gamma function, and ${}_2F_1(a,b,c;z)$ the hypergeometric function.
    \end{corollary}

    \begin{IEEEproof}
            Since $\mu(dh) = e^{-h}dh$, $L(sP_0)= \frac{1}{1 - sP_0}$ is well defined for $sP_0 < 1$, thus $s_\star = 1/P_0$, and $m_k = k!$. Thus, we get $A(s)$ by noticing that $\sum_{k=1}\frac{z^k}{k-\gamma} = \frac{z}{1-\gamma} {}_2F_1(1, 1-\gamma, 2-\gamma; z)$. Finally, we get $B(p)$ by noticing that $\int_z^\infty h^a e^{-h}dh = \Gamma(a+1, z)$.
    \end{IEEEproof}

    \begin{proposition}\label{prop:rateSpaceregulated}
        Same assumptions as Proposition~\ref{prop:fadingSpaceregulated}. $\Pr$-almost-surely, $\forall \theta > 0$,
        \begin{equation}\label{eq:rateconditionnedbypower}
            \Eo \left[ \exp\left\{ -\theta R(x, \Phi) \right\}\mid \Phi , P_x \right]  \leq  \exp\left\{-\zeta\left(\theta, \frac{P_x}{P_0}\right)\right\},
        \end{equation}
       where $\zeta(\theta ; p) = \sup_{\tilde{s} \in (0, P_0 s_\star)} \left\{ 
    \theta \log_2 \left( \frac{\tilde{s} p \ln 2}{\theta} \right) - 
    A(\tilde{s}) + B(p) \right.$ $\left. + \tilde{s}\gamma_0 + \ln \left[ \Gamma\left(\frac{\theta}{\ln 2}, \frac{\tilde{s} p }{1 + p/N_0}\right) - \Gamma\left(\frac{\theta}{\ln 2}, \tilde{s} p\right) \right]\right\}$, and $\Gamma(a,z) = \int_z^\infty t^{a-1}e^{-t}dt$ is the upper-incomplete Gamma function.
    \end{proposition}
    
    \begin{IEEEproof}
        See Appendix~\ref{proof:rateSpaceregulated}.
    \end{IEEEproof}



    

 \section{Delay analysis for CF networks under \texorpdfstring{\SIC}{SIC}}\label{sec:CF}

In this section, we investigate the benefits of a CF network incorporating the diversity gains that joint decoding provides. 
In CF networks, the receivers decode jointly the transmitted signals through a central server, which means that the received power of the signal of a transmitter $x\in\Phi$ to the network is now $P_0 \sum_{y\in\Psi} \ell(\|x-y\|)$. We suppose in the following that receivers cancel interference locally using SIC, meaning that the interference is the shot-noise of the interference from all receivers. This is reflected by the fact that, if the receiver $y\in\Psi$ decodes a transmission from $x\in\Phi$, the interference term is still the residual interference $P_0\sum_{w\in\Phi\setminus B(y, \|x-y\|)} \ell(\|w-y\|)$. When sent to the central server, the signal is jointly decoded using the signals from all receivers. The interference term is thus the superposition of the residual interference, namely $P_0\sum_{y\in\Psi} \sum_{w\in\Phi\setminus B(y, \|x-y\|)} \ell(\|w-y\|)$. Note that both the signal and interference power are not monotonic like in the previous section. Finally, we suppose that the receiver PP is not too sparse (this is controlled with void regulation) and that $\eta_0$ is not too small (as explained in the previous section, for one receiver, $\eta_0 > \frac{\sigma}{\alpha} - \frac{\rho^2 (\beta-2)}{8 \nu \alpha (\beta-1)^2}$). From these assumptions, we consider the interference-limited regime where the thermal noise term vanishes, and the \SINR~reduces to the Signal-to-residual-Interference Ratio (\SIR). 
The \SIR~of the typical receiver in this CF network is given by 
\begin{equation}
 \SIR(\Phi, \Psi) \triangleq \frac{\sum_{y\in\Psi} \ell(\|y\|)}{\sum_{y\in\Psi}\sum_{w\in\Phi\setminus B(y, \|y\|)} \ell(\|w-y\|)}.
\end{equation}
Thanks to the joint stationarity of $(\Phi, \Psi)$, $\SIR(\theta_{-x}\Phi, \theta_x\Psi)$ is the \SIR~of the transmitter located at $x\in\Phi$ and the network.
We suppose in the following that the receiver PP is jointly $(\sigma^\prime, \rho^\prime, \nu^\prime)$-ring regulated and $\tau$-void regulated. Nevertheless, as stated in \cite[Section~II.B]{feng2023spatial}, coupled with Lemma~\ref{lem:equivalence2}, in order for both regulations to be jointly possible, one needs $\tau < \nu^\prime$. Lattices and perturbed lattices are such PPs, c.f. \cite[Table~1]{feng2023spatial}.

\begin{proposition}\label{prop:CFSINRBound}
Let $(\Phi, \Psi)$ be a couple of PP jointly stationary. Let $\Psi$ be a $(\tau, K)$-covering and $(\sigma^\prime, \rho^\prime, \nu^\prime)$-ball regulated PP such that $\tau < \nu^\prime$. Let $\Phi$ be a $(\sigma, \rho, \nu)$-ring regulated PP w.r.t. $\Psi$. For an upper-bounded path-loss $\ell(r) = \Theta(r^{-\beta}), \ \beta > 4$, the \SIR~between the typical transmitter and the CF network is almost-surely lower-bounded, i.e.,
\begin{equation}\label{eq:CFSINRbound}
    \SIR(\Phi, \Psi) = \Omega\left(\frac{K}{C_\beta}\tau^{-\beta} \right), \quad \Pr-a.s.,
\end{equation}
where $C_\beta \triangleq \sigma\sigma' + \frac{2\sigma' \rho + \rho' \sigma}{\beta-1} + \frac{4\sigma' \nu + 2\nu' \sigma}{\beta-2}  + \frac{\rho' \rho}{(\beta-1)(\beta-2)} + \frac{2\rho' \nu}{(\beta-2)(\beta-3)}   + \frac{2\nu' \rho}{(\beta-1)(\beta-3)} + \frac{4\nu' \nu}{(\beta-2)(\beta-4)} .$
\end{proposition}

\begin{IEEEproof}
See Appendix~\ref{proof:CFSINRBound}.
\end{IEEEproof}

The difference between the cellular and cell-free networks regarding SIC, is that the decoding rate is no longer monotonic w.r.t. distance, since every receiver contributes in the decoding of every transmitter. Hence we can only provide a worst-case transmission analysis, in opposition to the typical transmission approach in the cellular network case. The \SIC-delay in CF network is thus given by the lower-bound in \eqref{eq:CFSINRbound}, if $\beta > 4$ and $\exists C >0$ such that $\eta_ 0 \geq C\frac{K}{C_\beta}\tau^{-\beta}$, then the worst-case delay \begin{equation}\label{def:CFdelay}
    D(\Phi, \Psi) \triangleq  \frac{m}{\log_2(1 + \SIR(\Phi, \Psi))}
\end{equation}
is such that $D(\Phi, \Psi) \leq \frac{m}{\log_2(1 + \eta_0)}, \quad \Pr^o_\Phi-a.s.,$ under the assumptions as Proposition~\ref{prop:CFSINRBound}.

\section*{Conclusions}

In this paper, we have established a novel framework utilizing spatial network calculus to derive decoding delay guarantees for both cellular and cell-free wireless networks. By shifting away from traditional Poisson assumptions and employing spatial regulation---specifically ring and void regulation---we successfully bounded the interference generated by concurrent transmissions. Our deterministic analysis demonstrates that the worst-case Successive Interference Cancellation (SIC) decoding delay is bounded by the span $T = \lceil\frac{m}{\log_2\left(1+\eta_0\right)}\rceil$. This formulation ensures that all transmissions originating within a specific coverage radius $\tau_0$ can be reliably decoded, and under $\tau_0$-void regulation, this latency guarantee scales across the entire network. Furthermore, our extension to cell-free architectures accommodates the non-monotonic relationship between decoding rate and distance, providing a robust worst-case transmission analysis based on the Signal-to-residual-Interference Ratio. 

While achieving strict deterministic guarantees necessitates stringent geometric constraints that can limit the coverage area, we have also expanded our framework to account for realistic channel impairments. Specifically, we provided an analytical outline for fading environments, demonstrating how stochastic network calculus can be leveraged to transition from strict deterministic limits to probabilistic tail bounds for decoding delays. Future research directions include extending this framework to accommodate dynamic, random traffic models---moving beyond the current assumption of a fixed message size of $m$ bits---and incorporating active cooperation among receivers in cell-free networks to fully harness macro-diversity benefits. Finally, future work should relax the assumptions of full network synchronization and static traffic to incorporate the temporal queueing dynamics that inherently affect deterministic delay metrics in multi-access environments.

\appendices

\section{Proof of Lemma~\ref{lem:equivalence2}}\label{proof:equivalence2}

For the first claim : Let the ring $\circledcirc_y(r, R)$ be covered by a rotating ball of radius $\delta = \frac{R-r}{2}$ centered at distance $\frac{R+r}{2}$ from $y=(y_1, y_2) \in \R^2$. Let $b_\theta = b\left( x, \frac{R-r}{2} \right)$ be the rotating ball, centered  at $x = \left( y_1 + \frac{R+r}{2}\cos\theta, y_2+\frac{R+r}{2}\sin\theta \right)$ for a given angle $\theta$. The angular width $\alpha_x$ of the ball from $y$ is given by $\sin\left(\frac{\alpha_x}{2}\right) = \frac{\delta}{d} = \frac{R-r}{R+r}$.
As a matter of scaling within the integration, since the integration is continuous, if the point $x$ is on the border of the ring, it will appear only once, which of Lebesgue measure $0$. For the integral $\int_0^{2\pi} \Phi(b_\theta) d\theta$ to be an upper bound for the discrete count $\Phi(\circledcirc_y(r, R))$, we require the angular coverage of each point to be at least $1$ radian ($\alpha_x \geq 1$). We utilize the Taylor decomposition of the sine function, i.e. $\sin\left(\frac{\alpha_x}{2}\right) = \frac{\alpha_x}{2} + O(\alpha_x^3)$. To find the thinnest ring we can get, we set $\alpha_x = 1$, thus 
$\frac{1}{2} \leq \frac{R-r}{R+r}$, which yields $3r \leq R$. The last two statement follow by taking $r=0$ in ring regulation.

\section{Proof of Lemma~\ref{lem:equivalence}}\label{appendix:equivalence}
    Same method than \cite[Theorem 1]{feng2023spatial}:
        Let $y\in\Psi$ and $R > r > 0$, $n\in\N, k=1,\dots,n, r_k=r+k\frac{R-r}{n}$ and $\o_k = \o_y(r,r_k)$. Then for the ring $\o_y(r,R)$ centered on $y$,
            \begin{align}
                &\sum_{x\in\Phi\cap \o_y(r,R)}  \ell(\|x-y\|)  \leq \sum_{k=0}^{n-1}    \ell(r_{k})(\Phi(\o_{k+1}) - \Phi(\o_k)) \label{lemma1:1}\\
                &\leq \sum_{k=1}^{n} (\ell(r_{k-1})-\ell(r_{k}))\Phi(\o_k) - \ell(R)\Phi(\o_y(r,R))\label{lemma1:2}\\
                &\leq  \sum_{k=1}^{n} (\ell(r_{k-1})-\ell(r_{k}))(\sigma + \rho (r_{k}-r) + \nu (r_{k}^2-r^2)\nonumber \\&\quad\quad - \ell(R)\Phi(\o_y(r,R)),\label{lemma1:3}
            \end{align} 
    where \eqref{lemma1:1} follows from the $\sigma$-additivity of $\Phi$ and the fact that $\ell$ is decreasing, \eqref{lemma1:2} is summation by parts, and \eqref{lemma1:3} follows from the ring-regulation of $\Phi$ and the fact that $\ell(r_{k-1})-\ell(r_{k}) >0$. The sum $\sum_{k=1}^{n} (\ell(r_{k-1})-\ell(r_{k}))(\sigma + \rho (r_{k}-r) + \nu (r_{k}^2-r^2))$ converges to the Stieltjes integral $-\int_r^R (\sigma + \rho (s-r) + \nu (s^2-r^2))d\ell(s)$ when $n\to\infty$, which exists because $\ell$ is bounded in $[r, R)$. Then we conclude by summation by parts: $-\int_r^R (\sigma + \rho (s-r) + \nu (s^2-r^2))d\ell(s) = \int_r^R  (\rho + 2\nu s)\ell(s)ds -\left[(\sigma + \rho (s-r) + \nu (s^2-r^2))\ell(s)\right]_r^R 
           = \sigma \ell(r)  + \rho \int_r^R \ell(s)ds + 2\nu\int_r^Rs\ell(s)ds - \ell(R)(\sigma + \rho(R-r) + \nu(R^2-r^2))$.  Since $\ell(R)(\sigma + \rho(R-r) + \nu(R^2-r^2) + \Phi(\o_y(r,R))$ is non-negative for $R > r$, we get the result. Furthermore, since $(\Phi, \Psi)$ is stationary, the latter holds for all $y \in \Psi$ almost surely.           
    For the converse, take $\ell = 1$.

\section{Proof of Lemma~\ref{lem:equivalence1}}\label{proof:equivalence1}
%

Let us first assume that $\Psi$ is $\tau$-void regulated. By the $\tau$-void regulation property, every point $x \in \R^2$ is at most at distance $\tau$ from at least one point in $\Psi$. It follows that the collection of disks $\{B(y, \tau)\}_{y \in \Psi}$ forms a covering of $\R^2$. By Kershner's theorem on the optimal covering density of the plane by congruent disks, the covering density satisfies $\theta \geq \frac{2\pi}{\sqrt{27}}$, see \cite{kershner1939number}. Let $x\in\R^2$ be any location on the Euclidean plane, and $n\in\N$.
Any ball $B(y, \tau)$ that intersects $B(x, (n-1)\tau)$ must have its center $y$ located within $B(x, n\tau)$ by the triangle inequality.  Therefore, the aggregate area of the covering disks centered in $B(x, n\tau)$ must satisfy $\sum_{y \in \Psi \cap B(x, n\tau)} \mathrm{Vol}(B(y, \tau)) \geq \theta \mathrm{Vol}(B(x, (n-1)\tau))$ leading to 
$\Psi(B(x, n\tau))  \pi \tau^2 \geq \frac{2\pi}{\sqrt{27}}  \pi (n-1)^2 \tau^2, \ \Pr$-almost-surely.
For $n \geq 2$, we use the bound $(n-1)^2 \geq \frac{1}{4}n^2$, yielding $\Psi(B(x, n\tau)) \geq \frac{\pi}{6\sqrt{3}} n^2$. Hence, $\Psi$ is $(\tau, \frac{\pi}{6\sqrt{3}})$-covering. 

For the converse, let us now suppose that $\Psi$ is $(\tau, \frac{\pi}{6\sqrt{3}})$-covering. Taking $n=1$, $\tau$-void regulation is reached with the integer ceiling $\lceil\frac{\pi}{6\sqrt{3}}\rceil = 1$.

\section{Proof of Proposition~\ref{prop:cumulativeCF}}\label{proof:cumulativeCF}

Let $B_n = \Psi(b(x, n\tau))$ be the number of receivers within distance $n\tau$ from $x\in\R^2$, and $C> 0$ such that $\ell(r) \geq Cr^{-\beta}$. By the $(\tau, K)$-covering property, $B_n \geq K n^2$ for all $n \in \N$. Let $f_n = (n\tau)^{-\beta}$. We express the total signal $S(x)=\sum_{y\in\Psi} \ell(\|x-y\|)$ by summing over the rings $\circledcirc_x(n\tau,(n+1)\tau)$, we get $S(x) \geq \lim_{N\to\infty} \sum_{n=1}^{N} (B_{n+1} - B_n)  C ((n+1)\tau)^{-\beta}$, yielding
\begin{align}
    &S(x) \geq -KC \tau^{-\beta} + \lim_{N\to\infty}B_{N+1}f_{N+1} - CB_1 f_1 \nonumber \\ &\quad\quad+ C\lim_{N\to\infty} \sum_{n=1}^{N}  B_n(f_n - f_{n+1})\label{eq:totalsignal1}\\
    &\geq -KC \tau^{-\beta} + C\sum_{n=1}^{\infty} B_n \tau^{-\beta} \left[ n^{-\beta} - (n+1)^{-\beta} \right] \label{eq:totalsignal2}\\
    &\geq C\frac{K}{\tau^\beta} \left( -1 + \sum_{n=1}^{\infty} n^2 \left[ n^{-\beta} - (n+1)^{-\beta} \right] \right),\label{eq:totalsignal3}
\end{align}
where \eqref{eq:totalsignal1} comes from the summation by parts formula $\sum_{n=1}^N (B_{n+1}-B_n) f_{n+1} = B_{N+1}f_{N+1} - B_1 f_1 + \sum_{n=1}^N B_n(f_n - f_{n+1})$,  \eqref{eq:totalsignal2} by taking the limit $N \to \infty$,  and \eqref{eq:totalsignal3} by substituting the lower bound $B_n \geq Kn^2$ because all terms in the sum are positive. Using the mean value theorem and the convexity of $n^{-\beta}$, $n^{-\beta} - (n+1)^{-\beta} \geq \beta (n+1)^{-(\beta+1)}$. The sum in \eqref{eq:totalsignal3}  behaves as $\sum_{n=1}^{\infty} n^2 (\beta (n+1)^{-\beta-1})-1 = \beta \sum_{n=1}^{\infty} n^{1-\beta} - 1$.
Since $\beta > 2$, the series $\sum n^{1-\beta}$ converges to a positive constant $D_\beta$. Thus $S(x) \geq \frac{K}{\tau^\beta} C( \beta D_\beta - 1 )$ and we conclude.

\section{Proof of Proposition~\ref{prop:coverage}}\label{proof:coverage}
Let $\sigma_\beta = \sigma$, $\rho_\beta = \frac{\rho}{\beta-1}$ and $\nu_\beta = \frac{2\nu}{\beta - 2}$.    
Using \eqref{eq:interferencebound}, we get a lower-bound on the \SINR~by lower bounding the signal-to-interference ratio (SIR) term and the signal-to-noise ratio $\frac{\ell(r)}{\gamma_0}$. Indeed $\SINR(x,\Phi) \geq \eta_0$ is implied by $\eta(r) \geq \eta_0$ which is itself implied by
\begin{equation}\label{eq:SIRbound}
\sigma_\beta  + \rho_\beta r + \nu_\beta r^2  \leq \frac{\alpha}{C\eta_0} \quad \text{and} \quad r^\beta \geq \frac{C\eta_0}{1-\alpha} \gamma_0,
\end{equation} for some $C > 0$,
c.f. \figurename~\ref{fig:rate}. For a \SINR~threshold $\eta_0$, a $(\sigma, \rho, \nu)$-ring regulated PP $\Phi$, using \eqref{eq:SIRbound} and the fact that $\ell(r) = \Theta(r^{-\beta})$, we know that there exists a constant $C > 0$ such that $\tau(\theta_y\Phi) \geq \tau_0, \Pry-a.s.$, where $\tau_0$ is given in Proposition~\ref{prop:coverage}.
 Since $\tau_0$ is a lower-bound on the coverage radius of the typical receiver, the \SINR~of any transmitter from distance at most $\tau_0$ from any receiver decoded with \SIC~is greater than $\eta_0$ with probability $1$. Eq. \eqref{eq:SIRbound} admits a solution $\tau_0 = \min\left\{ \frac{\sqrt{\rho_\beta^2 + 4\nu_\beta\left(\frac{\alpha}{C\eta_0} - \sigma_\beta\right)} - \rho_\beta}{2\nu_\beta}, \left(\frac{1-\alpha}{C\eta_0 \gamma_0}\right)^{1/\beta}\right\}$ only if $\eta_0 >\frac{\sigma_\beta}{\alpha} - \frac{\rho_\beta^2}{4\nu_\beta\alpha}$.
 
\section{Proof of Proposition~\ref{prop:uplinkdelay}}\label{appendix:sicdecoding delay}
        Since $d(w-y, \theta_y\Phi) \leq  \frac{m}{\log_2\left(1+\eta(\|w-y\|)\right)}$, $\Pry$-almost-surely, and from the ergodicity of  $(\Phi, \Psi)$ and the properties of Palm probabilities, we have that
    \begin{equation}\label{eq:worst_case_bounded_by_typical}
    D(x,\Phi) \leq \sup_{w\in \Phi \cap b(y, \|x-y\|)}  \frac{m}{\log_2\left(1+\eta(\|w-y\|)\right)},
    \end{equation} $\Pry$-almost-surely. Note that, in order to have $\tau_0 > 0$, we need $\eta_0 > 1/\sigma_\beta$.  Finally we conclude with \eqref{eq:delaySRbound}, the fact that $\eta(r)$ is decreasing w.r.t. $r$ and $\eta(\tau_0) = \eta_0$. For the supremum $\sup_{x\in\Phi_0} D(x,\Phi)$: using \eqref{eq:worst_case_bounded_by_typical}, we just need to show that $\Pr\left(\Pro\left(\sup_{x\in\Phi\cap b(o, \tau_0)} d(x, \Phi)  > \frac{m}{\log_2\left(1+\eta_0\right)}\mid\Phi\right) = 0\right)=1.$ The supremum in $\sup_{x\in\Phi\cap b(o, \tau_0)} \frac{m}{\log_2(1 + \eta(\|x\|))}$ is reached for $\|x\| = \tau_0$. Finally we conclude with \eqref{eq:delaySRbound}, the fact that $\eta(r)$ is decreasing w.r.t. $r$ and $\eta(\tau_0) = \eta_0$.

\section{Proof of Lemma~\ref{lem:SNCvirtualdelay}}\label{proof:SNCvirtualdelay}%
$\Pro$-almost-surely,
\begin{align} &\Pro(d(x, \Phi) > T \mid \Phi, P_x)\nonumber\\ &= \Pro(\sup_{0 \leq t \leq T} \frac{m}{t} - R(x,\Phi) > 0   \mid \Phi)\label{eq:delaysup}\\
               &= \sum_{t=0}^T \Pro( m > tR(x,\Phi)  \mid \Phi) \label{eq:unionbound}\\
               &\leq \exp(m\theta) \sum_{t=0}^T \exp(-t \theta R(x,\Phi)) \mid \Phi) \label{eq:chernov}\\
                &\leq \exp(m\theta - T \zeta(\theta; P_x))) \sum_{t=0}^\infty \exp(-t \zeta(\theta; P_x))) \label{eq:servicebound1}\\
                &\leq \frac{\exp(m\theta - T \zeta(\theta; P_x)))}{1 - \exp(- \zeta(\theta; P_x)))} \label{eq:servicebound2}
           \end{align}
           where \eqref{eq:delaysup} comes with the definition of delays, \eqref{eq:unionbound} the union-bound, \eqref{eq:chernov} Chernov's inequality, \eqref{eq:servicebound1} from \eqref{eq:defRLregulation} and \eqref{eq:servicebound2} because $\exp(- \zeta(\theta; P_x))) < 1$.

\section{Proof of Proposition~\ref{prop:fadingSpaceregulated}}\label{proof:fadingSpaceregulated}

$\Pro$-almost-surely
    \begin{align}
        &\ln \Eo\left[ \exp\left\{ sI_x \right\} \mid \Phi, P_x\right] \nonumber \\
        &= \sum_{w\in\Phi} \ln \Eo \left[ \exp \left\{s P_0 h_w\ell(\|w\|)\mathbf{1}(h_w \leq \frac{P_x}{ P_0\ell(\|w\|)}) \right\} \mid P_x\right]  \label{eq:interference1}\\
        &= \sum_{w\in\Phi} \ln\left(\int_0^{\frac{P_x}{ P_0 \ell(\|w\|)}} \exp \left\{s P_0 h\ell(\|w\|) \right\} \mu(dh)+ \int_{\frac{P_x}{P_0 \ell(\|w\|)}}^\infty  \mu(dh)  \right)\label{eq:interference2}
    \end{align}
    where \eqref{eq:interference1} comes from the fact that the $(h_w)_w$ are i.i.d. w.r.t.  the distribution $\mu$, and \eqref{eq:interference2} because $\exp(a\mathbf{1}_{A}) = \exp(a)\mathbf{1}_{A} + \mathbf{1}_{A^c}$.  In order to apply shot-noise ring regulation as defined in \eqref{eq:shotnoise} for $r = 0$ and $R=\infty$, we need the function $\tilde\ell(r) = \ln\left(\int_0^{\frac{P_x}{ P_0 \ell(r)}} e^{sP_0\ell(r)h} \mu(dh)+ \int_{\frac{P_x}{P_0 \ell(r)}}^\infty  \mu(dh)  \right)$ to be non-negative, bounded and non-increasing w.r.t.  $r \geq 0, \forall s \in [0, s_\star]$. The boundness is inherited by the boundness and positiveness of $\ell$. It is positive since $\exp \left\{ sh\ell(r \right\})$ is greater than $1$ for $s>0$. 
    It decreases w.r.t.  $r$ since $\ell$ decreases with $r$, so as the interval $\left[\frac{h\ell(\|x\|)}{\ell(\|w\|)},\infty\right) $. 
    
    Let us first show that $\tilde\ell$ has the right integration property. As $r\to\infty$, $\int_0^{\frac{p}{ P_0 \ell(r)}} e^{sP_0\ell(r)h} \mu(dh)  = \int_0^{\frac{p}{ P_0 \ell(r)}} \mu(dh) + s\ell(r)\int_0^{\frac{p}{ P_0 \ell(r)}} h \mu(dh) + o(\ell(r))$, and thus $\tilde\ell(r) \leq \ln \left( s\ell(r)\int_0^{\frac{p}{ P_0 \ell(r)}} h \mu(dh) + o(\ell(r)) + 1\right)$, yielding
    \begin{equation}\label{eq:taylor}
        \tilde\ell(r) \leq  s\ell(r)m_1 + o(\ell(r)),
    \end{equation}
     since $\ln(1+z) \leq z$ for small $z\geq0$. With \eqref{eq:taylor}, we get $\int r\tilde\ell(r) dr < \infty $ from $\int r\ell(r) dr < \infty $.
    
    Let $Q(z) = \int_z^\infty \mu(dh)$. Now let us remark that for $\hat\ell(r) = \ln \left( L(sP_0\ell(r)) + Q(\frac{p}{P_0\ell(r)})\right)$, \begin{equation}\tilde\ell(r) \leq \hat\ell(r) \leq L(sP_0\ell(r))-1 + Q(\frac{p}{P_0\ell(r)}).\label{eq:boundLQ}\end{equation} We conclude with Lemmas~\ref{lem:Alemma} and \ref{lem:Blemma}.  Now using that $\Phi$ is $(\sigma, \rho, \nu)$ regulated, 
     \begin{equation}\label{eq:boundAB}
         \sum_{w\in\Phi} \hat\ell(\|w\|) \leq \sigma \hat\ell(0) + \rho\int_0^\infty \hat\ell(r)dr + 2\nu \int_0^\infty r\hat\ell(r)dr, 
     \end{equation}
        $\Pr$-almost-surely, if $s \in (0, s_\star)$.

    \begin{lemma}\label{lem:Alemma} For $\ell(r) = \max\{1,Cr\}^{-\beta}$, $\forall s \in(0,s_\star)$,
        \begin{multline}
            A(z) = \left(\sigma + \frac{\rho}{C} + \frac{\nu}{C^2}\right) (L(z)-1) + \frac{\rho}{C\beta} \sum_{k=1}^\infty \frac{z^km_k }{k!(k - 1/\beta)} \\+ \frac{\nu}{C^2\beta} \sum_{k=1}^\infty \frac{z^km_k}{k!(k - 2/\beta)}.
        \end{multline}
    \end{lemma}
    
    \begin{IEEEproof}[Proof of Lemma~\ref{lem:Alemma}]
    Integrating the first term leads to the upper-bound in \eqref{eq:boundLQ} leads to $\sigma (L(sP_0\ell(0))-1) + \rho \int_0^\infty (L(sP_0\ell(r))-1)dr + 2\nu \int_0^\infty r (L(sP_0\ell(r))-1)dr.$ Expanding the m.g.f. in Taylor series leads to $L(sP_0\ell(r))-1 = \E\left[ \sum_{k\geq 1} \frac{(sP_0\ell(r)h)^k}{k!} \right].$ Now let us consider that $\ell(r) = \max\{1,Cr\}^{-\beta}, \beta > 2$. Hence, $L(sP_0\ell(r))-1 = \sum_{k\geq 1} \frac{s^kP_0^k\max\{1,Cr\}^{-k\beta} \E[h^k]}{k!}$. Thus with Fubini's theorem we get, 
    \begin{align*}
        &\int_0^\infty (L(sP_0\ell(r))-1)dr \\&= \sum_{k\geq 1} \frac{s^kP_0^k  m_k}{k!}\int_0^{1/C}dr +  \sum_{k\geq 1} \frac{s^kP_0^kC^{-k\beta} m_k}{k!} \int_{1/C}^\infty r^{-k\beta}dr\\
        &= \frac{1}{C}\sum_{k\geq 1} \frac{s^kP_0^k m_k}{k!} +  \sum_{k\geq 1} \frac{s^kP_0^k C^{-k\beta} m_k}{k!}  \frac{1}{C^{1-k\beta}(k\beta-1)} \\
        &= \frac{L(sP_0)-1}{C} +  \frac{1}{C}\sum_{k\geq 1} \frac{s^kP_0^k  m_k}{k!(k\beta-1)}   
    \end{align*}
and with the same reasoning, $\int_0^\infty r(L(sP_0\ell(r))-1)dr = \frac{L(sP_0)-1}{2C^2} +  \frac{1}{2C^2}\sum_{k\geq 1} \frac{s^kP_0^k  m_k}{k!(k\beta-2)}$.  
    \end{IEEEproof}

    \begin{lemma}\label{lem:Blemma}
        $B(p) = \sigma \mu\left(h \in \R_+ : h > p\right) +  \frac{\rho }{p^{1/\beta}} \int_{p}^\infty h^{1/\beta}\mu(dh)+ \frac{\nu }{p^{2/\beta}} \int_{p}^\infty h^{2/\beta}\mu(dh)$. 
    \end{lemma}

    \begin{IEEEproof}[Proof of Lemma~\ref{lem:Blemma}]
           Integrating the second term in \eqref{eq:boundLQ} leads to 
    $\rho\int_0^\infty Q\left(\frac{p}{P_0\ell(r)}\right)dr =  \rho\left(\frac{P_0}{p}\right)^{1/\beta} \int_{p/P_0}^\infty h^{1/\beta}\mu(dh),$
    and
    $2\nu\int_0^\infty r Q\left(\frac{p}{P_0\ell(r)}\right)dr = \nu\left(\frac{P_0}{p}\right)^{2/\beta} \int_{p/P_0}^\infty h^{2/\beta}\mu(dh).$
    \end{IEEEproof}

    \begin{IEEEproof}[End of proof of Proposition~\ref{prop:fadingSpaceregulated}]
        Conclude with Lemmas~\ref{lem:Alemma}, \ref{lem:Blemma} and \eqref{eq:boundAB} and \eqref{eq:boundLQ}.
    \end{IEEEproof}

\section{Proof of Proposition~\ref{prop:rateSpaceregulated}}\label{proof:rateSpaceregulated}

    \begin{lemma}\label{lem:SrINRbound}
         $\Pr$-almost-surely, \begin{multline}
             \Pro(\SINR(x,\Phi) \leq t \mid \Phi, P_x) \\ \leq \exp\left\{A(sP_0) + B(\frac{P_x}{P_0}) -s\left(\frac{P_x}{t} - N_0 \right)\right\}  \mathbf{1}(\frac{P_x}{N_0} > t).
         \end{multline}
    \end{lemma}

    \begin{IEEEproof}[Proof of Lemma~\ref{lem:SrINRbound}]
    For all $s\in[0,s_\star]$,
        \begin{align}
         &\Pro(\SINR(x,\Phi) \leq t \mid \Phi, P_x) \nonumber\\&=  \Pro\left(I_x   \geq \frac{P_x}{t} - N_0, P_x  > tN_0\mid \Phi, P_x\right) \nonumber  \\
         &\leq  \frac{\Eo\left[\exp\left\{sI_x  \right\} \mid \Phi, P_x\right]\mathbf{1}(P_x > tN_0)}{\exp\left\{s\frac{P_x}{t} - sN_0\right\}} \label{eq:sinrcdfbound1}\\
          &\leq  \exp\left\{A(sP_0) + B(\frac{P_x}{P_0}) -s\left(\frac{P_x}{t} - N_0 \right)\right\}\mathbf{1}(\frac{P_x}{N_0} > t)  \label{eq:sinrcdfbound2}
    \end{align} 
    with \eqref{eq:sinrcdfbound1} Chernoff's inequality and \eqref{eq:sinrcdfbound2} Proposition~\ref{prop:fadingSpaceregulated}.
    \end{IEEEproof}

\begin{IEEEproof}[Proof of Proposition~\ref{prop:rateSpaceregulated}]
      Let $s \in [0, s_\star]$. Then,
    \begin{align}
         &\Eo \left[ \exp\left( -\theta R(x, \Phi) \right)\mid \Phi, P_x=p\right] \nonumber\\ 
         &=  \int_0^1  \Pro(\exp\left( -\theta \log_2(1 + \SINR(x, \Phi)) > t \mid \Phi, P_x=p\right)dt\label{eq:ratebound4}\\
            & \leq \int_0^1 \exp\left\{A(sP_0) + B(p/P_0) -s\left(\frac{p}{t^{-\ln2/\theta}-1} - N_0 \right)\right\} \nonumber\\&\quad\quad\quad\times \mathbf{1}(p> (2^{-\ln t/\theta}-1)N_0)dt\label{eq:ratebound2}\\
            &\leq \exp\left\{ A(sP_0) + B(p/P_0) + sN_0 + \ln \int_{e^{-\theta \log_2(1 + p/N_0)}}^1\!\!\!\!\!\!\!\!\!\!\! e^{ -spt^{\ln2/\theta} }dt\right\}\label{eq:ratebound}
    \end{align} 

    where \eqref{eq:ratebound4} is from the fact that for a positive r.v.$X$, $\Pr(\exp \left\{ -\theta X \right\} > t) = 0$ for all $\theta >0$, $t > 1$, \eqref{eq:ratebound2} from Lemma~\ref{lem:SrINRbound}, and from $2^{-\ln(t)/\theta} = t^{-\ln2/\theta}$ and $2^{-\ln(t)/\theta} >2$ for $t\in(0,1)$, \eqref{eq:ratebound} because $\int_0^1\exp\left(-\frac{sp}{t^{-\ln2/\theta}-1 }\right)dt$ is finite and upper-bounded by $\int_0^1\exp\left(-spt^{\ln2/\theta}\right)dt$.

    With the change of variable $u=spt^{\ln2/\theta}$, we get
   \begin{multline}
        \int_{e^{-\theta \log_2(1 + \frac{p}{N_0})}}^1  \exp \left\{ - spt^{\ln2/\theta} \right\}dt  \\ = \int_{\frac{sp}{\log_2(1 + \frac{p}{N_0})}}^{sp} e^{-u} (sp)^{-\frac{\theta}{\ln2}}u ^{\frac{\theta}{\ln2}-1}du \\
        = \frac{\theta}{\ln2} \frac{\Gamma(\frac{\theta}{\ln2}, \frac{sp}{\log_2(1 + \frac{p}{N_0})}) - \Gamma(\frac{\theta}{\ln2}, sp)}{(sp)^{\frac{\theta}{\ln2}} }\label{eq:gammafunc}
    \end{multline}
     where $\Gamma(a,z)$ is the upper-incomplete Gamma function. Let $\gamma_0 = N_0/P_0$ and $e(p) =  \frac{\log_2(1 +  p/\gamma_0)}{p}$. Taking the log and the infinimum over $s \in [0,s_\star]$ in \eqref{eq:gammafunc} leads to the result.
\end{IEEEproof}

\section{Proof of Proposition~\ref{prop:CFSINRBound}}\label{proof:CFSINRBound}
From Proposition~\ref{prop:cumulativeCF}, we get that the aggregate signal is lower-bounded, i.e., $$ \sum_{y\in\Psi} \ell(\|y\|)=\Omega(K\tau^{-\beta}), \quad \Pr-a.s.$$ The bound on the interference term is twofold: firstly from Proposition~\ref{prop:shotnoisenofading}, $\forall y\in\Psi$, $$\sum_{w\in\Phi\setminus B(y, r)} \ell(\|w-y\|) = O\left(\sigma r^{-\beta} + \frac{\rho}{\beta-1} r^{1-\beta} + \frac{2\nu}{\beta-2} r^{2-\beta}\right),$$ secondly from the ball regulation of $\Psi$, and more specifically the shot-noise regulation induced by ring regulation, c.f. Lemma~\ref{lem:equivalence}, we get for a transmitter located at the origin, i.e., $\Pr$-almost-surely, for $\beta > 4$, since we integrate $r^{3-\beta}$,
\begin{multline*}
\sum_{y \in \Psi} \sum_{w \in \Phi \setminus B(y, \|y\|)} \ell(\|w-y\|) \leq \sigma' \left( \sigma + \frac{\rho}{\beta-1} + \frac{2\nu}{\beta-2} \right) \\
+ \left[ \frac{\sigma' \rho + \rho' \sigma}{\beta-1} + \frac{\rho' \rho}{(\beta-1)(\beta-2)} + \frac{2\rho' \nu}{(\beta-2)(\beta-3)} \right] \\
+ \left[ \frac{2(\sigma' \nu + \nu' \sigma)}{\beta-2} + \frac{2\nu' \rho}{(\beta-1)(\beta-3)} + \frac{4\nu' \nu}{(\beta-2)(\beta-4)} \right].
\end{multline*}

\bibliographystyle{plain}
\bibliography{biblio}

\end{document}